\documentclass[pra,twocolumn,english,reprint, superscriptaddress, breaklinks=true, showkeys, showpacs=false, nofootinbib]{revtex4}

\usepackage[T1]{fontenc}
\usepackage{color}
\usepackage{babel}
\usepackage{float}
\usepackage{amssymb}
\usepackage{listings}
\usepackage{graphicx}
\usepackage{subcaption}
\usepackage{amsmath}
\usepackage{amsthm}
\usepackage{svg}
\usepackage{caption}
\newtheorem{theorem}{Theorem}
\newtheorem{lemma}[theorem]{Lemma}
\newtheorem{proposition}[theorem]{Proposition}
\newtheorem{definition}{Definition}
\newtheorem{property}[definition]{Property}
\usepackage[unicode=true,pdfusetitle, bookmarks=true,bookmarksnumbered=false,bookmarksopen=false, breaklinks=true,pdfborder={0 0 0},backref=false,colorlinks=true]{hyperref}
\usepackage{breakurl}
\definecolor{Code}{rgb}{0,0,0}
\definecolor{Decorators}{rgb}{0.5,0.5,0.5}
\definecolor{Numbers}{rgb}{0.5,0,0}
\definecolor{MatchingBrackets}{rgb}{0.25,0.5,0.5}
\definecolor{Keywords}{rgb}{0,0,1}
\definecolor{self}{rgb}{0,0,0}
\definecolor{Strings}{rgb}{0,0.63,0}
\definecolor{Comments}{rgb}{0,0.63,1}
\definecolor{Backquotes}{rgb}{0,0,0}
\definecolor{Classname}{rgb}{0,0,0}
\definecolor{FunctionName}{rgb}{0,0,0}
\definecolor{Operators}{rgb}{0,0,0}
\definecolor{Background}{rgb}{0.98,0.98,0.98}
\lstdefinelanguage{Python}{
numbers=left,
numberstyle=\footnotesize,
numbersep=1em,
xleftmargin=1em,
framextopmargin=2em,
framexbottommargin=2em,
showspaces=false,
showtabs=false,
showstringspaces=false,
frame=l,
tabsize=4,
basicstyle=\ttfamily\small\setstretch{1},
backgroundcolor=\color{Background},
commentstyle=\color{Comments}\slshape,
stringstyle=\color{Strings},
morecomment=[s][\color{Strings}]{"""}{"""},
morecomment=[s][\color{Strings}]{'''}{'''},
morekeywords={import,from,class,def,for,while,if,is,in,elif,else,not,and,or,print,break,continue,return,True,False,None,access,as,,del,except,exec,finally,global,import,lambda,pass,print,raise,try,assert},
keywordstyle={\color{Keywords}\bfseries},
morekeywords={[2]@invariant,pylab,numpy,np,scipy},
keywordstyle={[2]\color{Decorators}\slshape},
emph={self},
emphstyle={\color{self}\slshape},
}
\makeatletter
\@ifundefined{textcolor}{}
{%
 \definecolor{BLACK}{gray}{0}
 \definecolor{WHITE}{gray}{1}
 \definecolor{RED}{rgb}{1,0,0}
 \definecolor{GREEN}{rgb}{0,1,0}
 \definecolor{BLUE}{rgb}{0,0,1}
 \definecolor{CYAN}{cmyk}{1,0,0,0}
 \definecolor{MAGENTA}{cmyk}{0,1,0,0}
 \definecolor{YELLOW}{cmyk}{0,0,1,0}
}
\hypersetup{colorlinks=true,citecolor=blue,linkcolor=cyan,urlcolor=blue,filecolor= green, breaklinks=true}
\usepackage{url}
\usepackage{breakurl}

\makeatother

\begin{document}

\title{Using quantum computers to identify prime numbers via entanglement dynamics}

\author{Victor F. dos Santos\href{https://orcid.org/0009-0009-0319-4852}{\includegraphics[scale=0.05]{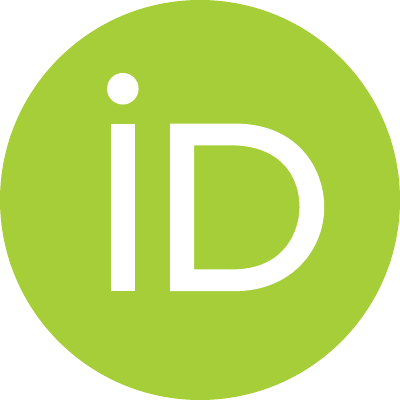}}}
\email[]{victorfds997@gmail.com}
\affiliation{Physics Department, Center for Natural and Exact Sciences, Federal University of Santa Maria, Santa Maria, RS, 97105-900, Brazil}

\author{Jonas Maziero\href{https://orcid.org/0000-0002-2872-986X}{\includegraphics[scale=0.05]{orcidid.pdf}}}
\email[]{jonas.maziero@ufsm.br}
\affiliation{Physics Department, Center for Natural and Exact Sciences, Federal University of Santa Maria, Santa Maria, RS, 97105-900, Brazil}

\begin{abstract}
Recently, the entanglement dynamics of two harmonic oscillators initially prepared in a separable-coherent state was demonstrated to offer a pathway for prime number identification. This article presents a generalized approach and outlines a deterministic algorithm making possible the implementation of this theoretical concept on scalable fault-tolerant qubit-based quantum computers. We prove that the diagonal unitary operations employed in our algorithm exhibit a polynomial-time complexity of degree two, contrasting with the previously reported exponential complexity of general diagonal unitaries.
\end{abstract}


\maketitle

\section{Introduction}

The quest to reliably and efficiently identify prime numbers (PNs) remains a topic of great interest in number theory \cite{Bressoud1989,crandall,Baillie2021,Granville2005}, particularly due to its intriguing connection with the non-trivial zeros of Riemann's zeta function \cite{schmayer, wolf,Feiler2013,Sierra2008}. Over the centuries, numerous classical algorithms have been devised for identifying primes, each offering its own set of advantages and limitations \cite{Bressoud1989,crandall, aaronson}. Among these, the AKS primality test stands out as the first deterministic algorithm to exhibit polynomial-time complexity for verifying the primality of individual integers, albeit with a polynomial degree that renders it less efficient for larger numbers \cite{AKS}. Conversely, the Sieve of Eratosthenes
offers a simpler approach, focusing on identifying all PNs within a specified range \( N \). Its time complexity, \( \mathcal{O}(N \log \log N) \), renders it particularly efficient for this purpose
\cite{crandall}.

While classical algorithms for PNs identification \cite{Sieve2017, Miller1976, Adleman1983, AKS} have undergone significant development, their adaptation to the realm of quantum computers (QCs) remains relatively limited
\cite{Donis-Vela2017, Chau1997, Li2012}. However, the intersection of such questions with experimental physics \cite{Garcia-Martin2020, Mussardo2020,alexandre} presents a promising avenue for the development of more intuitive quantum algorithms. A notable recent study \cite{alexandre} proposed an innovative approach to primality testing using quantum optics. In their work, researchers devised an experiment involving the entanglement of two quantum harmonic oscillators initially prepared in coherent states, followed by the measurement of the reduced linear entropy of one of them. 
They theorized that information regarding PNs could be extracted from the Fourier modes of the reduced linear entropy: PNs were expected to adhere to a lower bound curve, while composite numbers would consistently surpass this bound. Although the experimental implementation has yet to be realized, and has known scalability limitations, their theoretical groundwork has laid the foundation for us to generalize their approach and to develop a deterministic algorithm tailored for implementation on qubit-based QCs.

In this article, we build upon the theoretical framework proposed in Ref. \cite{alexandre}, aiming to adapt it for implementation on qubit-based QCs by removing certain restrictions imposed on the Hamiltonian and initial states. As a result, we demonstrate, as detailed in the Appendices, that the class of diagonal unitary gates utilized in our approach can be implemented in polynomial time, contrary to the expectations set forth in Refs. \cite{Bullock2004,welch}. Our algorithm is designed to determine all PNs within a given range \( N \) through the manipulation of a bipartite system \( AB \) and the measurement of the linear entropy of entanglement \cite{Bennett1996,Vidal1999,Basso2022,Scherer2021} of subsystem \( A \) over a period \( T \).

Our methodology unfolds with the following steps. Firstly, we modify the definitions to align with the peculiarities of qubit-based QCs. Secondly, we select a suitable initial state that can be efficiently prepared. Thirdly, we efficiently prepare an evolved state using the techniques outlined in Ref. \cite{welch}, which surprisingly results in 
exponential gate cost reduction
in comparison to the general case. Subsequently, we measure the reduced purity, a task that can be executed efficiently \cite{ekert}. Following this, we calculate the Fourier modes of the reduced purity function via numerical integration methods \cite{Press2007}.

Given a dataset encompassing all points within the range \( N \), our algorithm enables the deterministic identification of Fourier modes corresponding to PNs, allowing for the distinction between primes and composites. We quantify the number of gates utilized at each step, with a specific focus on \( Z \)-rotations, Controlled-NOT and Hadamard gates. Additionally, we discuss simulations conducted using Qiskit \cite{Qiskit} and explore potential enhancements to our algorithm for more efficient implementation on real qubit-based QCs.

We begin by establishing key definitions.
Let \( A \) and \( B \) represent the respective subsystems, each characterized by a time-independent Hamiltonian \( \hat{H}_A \) and \( \hat{H}_B \), where \( \hat{H}_A = \hat{H}_B \). We define a bipartite Hamiltonian \( \hat{H}_{AB} = \lambda \hat{H}_A \otimes \hat{H}_B \), with \( \lambda \in \mathbb{R} \) denoting the coupling constant. The corresponding time-evolution operator is given by \( \hat{U}(t) = e^{-i\hat{H}_{AB}t/\hbar} \) \cite{Nielsen2000}.

To obtain a distinction between prime and composite numbers, we employ the initial state 
\( |\psi(0)\rangle_{AB}=|\phi\rangle_A\otimes|\phi\rangle_B \). A suitable choice for these individual states is 
\begin{equation}
\label{phi_state}
|\phi\rangle_S = \sum_{n_S = 1}^d c_{n_S}|E_{n_S}\rangle,
\end{equation}
where \( S = A,B \) represents the subsystem index, \( d \) is the dimension of each subsystem, \( c_{n_S} \neq 0 \) are the initial state coefficients, and \( \{E_{n_S}\}_{n_S = 1}^d \) and \( \{|E_{n_S}\rangle\}_{n_S = 1}^d \) denote the eigenvalues and eigenvectors of each subsystem Hamiltonian, respectively. The evolved state at time \( t \) is 
\( |\psi(t)\rangle_{AB} = \hat{U}(t)|\psi(0)\rangle_{AB} = \sum_{n_A, n_B = 1}^d c_{n_A}c_{n_B}e^{-i\lambda E_{n_A}E_{n_B} t/\hbar}|E_{n_A}E_{n_B}\rangle \).

Our main condition requires that the energy levels of both individual Hamiltonians are equidistant, i.e., 
\( E_{n_S} =  n_S\mu \) for some constant \( \mu \in \mathbb{R} \). Defining \( \omega = \lambda \mu^2/\hbar \), we find that
\begin{equation}
\label{unitary}
\hat{U}(t) = \sum_{n_A, n_B = 1}^d e^{-i\omega n_An_Bt}|E_{n_A}E_{n_B}\rangle\langle E_{n_A}E_{n_B}|,
\end{equation}
and
\begin{equation}
\label{psi_t}
|\psi(t)\rangle_{AB} = \sum_{n_A, n_B = 1}^d c_{n_A}c_{n_B}e^{-i\omega n_{A}n_{B} t}|E_{n_A}E_{n_B}\rangle.
\end{equation} 

A key result of our research is the demonstration of high gate-efficiency for implementing the diagonal unitary gate specified in Eq. (\ref{unitary}), as detailed in the Appendix \ref{sec:appendixD}. We show that the gate cost for constructing the \( q \)-qubit unitary gate \( \hat{U}(t) \) using this method is a polynomial function \( G_2(q) = \frac{3}{4}q^2 + q \). This result not only facilitates PNs identification but also paves the way for efficient implementation of similar unitary gates in future qubit-based QCs research. 

The remainder of this article is organized as follows. In Sec. \ref{sec:purity}, we give the general expression for the reduced purity of a subsystem \( A\) and highlight its mathematical properties. Next, in Sec. \ref{sec:fourier}, we establish the theoretical connection between the Fourier modes of the reduced purity and the distribution of prime numbers. Then, in Sec. \ref{sec:algorithm}, we report our quantum algorithm, specifying the techniques used and the associated computational costs. In Sec. \ref{sec:simulations}, we present the results of simulations made using Qiskit. Finally, we conclude in Sec. \ref{sec:conclusions} by revisiting the key points of our method and quantum algorithm, while discussing limitations and further potential improvements with respect to an implementation on quantum hardware.
Additional details supporting our findings are provided in the appendices. Appendix \ref{sec:appendixA} is a summary of the technique developed in Ref. \cite{welch} for the implementation of general diagonal unitary gates using Walsh functions. Appendix \ref{sec:appendixB} presents a proof for a known identity that relates tensor products of Pauli $\hat{Z}$ gates and $\widehat{CNOTs}$, and a proof for how this identity relates to the implementation of exponentials of Walsh operators. In Appendix \ref{sec:appendixC}, we prove some results regarding Walsh matrices and delineate our notation for them, as it will be heavily used in further demonstrations. Then, Appendix \ref{sec:appendixD} uses results from the previous appendices to rigorously demonstrate that the diagonal unitary gate in Eq. (\ref{unitary}) may be implemented efficiently using only a polynomial number (with respect to the number of qubits) of elementary gates. Furthermore, Appendix \ref{sec:appendixE}
is a direct proof for a modified version of the SWAP test, aiming for the estimation of the reduced purity.

\section{Reduced purity}
\label{sec:purity}

Without loss of generality, we designate subsystem \( A \) for computing the reduced purity \( \gamma_A(t) \). Let us begin by revisiting the definition of the reduced density operator \( \hat{\rho}_{A}(t) \) for a system \( AB \) with density operator \( \hat{\rho}_{AB}(t) \), given as \( \hat{\rho}_A(t) = \operatorname{Tr}_{B}(\hat{\rho}_{AB}(t)) \), where \( \operatorname{Tr}_B(.) \) denotes the partial trace function \cite{Maziero2017} over subsystem \( B \). The reduced purity function, \( \gamma_A(t) = \operatorname{Tr}(\hat{\rho}^2_A(t)) \), can then be computed from \( \hat{\rho}_A(t) \). This quantity is related to the linear entanglement entropy by $E_l(|\psi(t)\rangle_{AB})=1-\gamma_A(t).$
Given that \( \hat{\rho}_{AB}(t) = |\psi(t)\rangle_{AB}\langle\psi(t)|\), the reduced purity can be straightforwardly expressed as
\begin{equation}
\label{purity}
    \gamma_{A}(t) = \sum_{j, k, l, m = 1}^d |c_{j}|^2|c_{k}|^2|c_{l}|^2|c_{m}|^2e^{-i \omega t(j - k)(l - m) }.
\end{equation}

It is noteworthy to highlight several properties of the function \( \gamma_A(t) \) defined in Eq. (\ref{purity}). Firstly, it exhibits time periodicity with period \( T = 2\pi/\omega \), a characteristic stemming directly from the time evolution of our system. 
Secondly, a notable observation arises from the structure of the sum in Eq. (\ref{purity}): the indices \( j, k, l, \) and \( m \) all take the same values. As a consequence, the imaginary parts of the phases \( e^{-i \omega t(j - k)(l - m) } \) for a fixed \( t \) mutually cancel each other. This cancellation is crucial, ensuring that \( \gamma_{A}(t) \) remains a real-valued function. 
This function is symmetric about half the period, \( \gamma_{A}(T/2 + h) = \gamma_{A}(T/2 - h) \), which enables us to halve the number of times we need to execute the quantum circuit to obtain it.


\section{Mapping prime numbers with Fourier modes} 
\label{sec:fourier}

The reduced purity function given by Eq. (\ref{purity}) can be expressed as a finite sum of cosines, where the maximum number of Fourier modes \( \alpha_n \) is \((d-1)^2\). Therefore, employing a Fourier expansion in this scenario yields:
\begin{equation}
\gamma_{A}(t) = \alpha_0 + \sum_{n = 1}^{(d-1)^2}\alpha_n\cos(n \omega t),
\end{equation}
where \( \alpha_0 \) represents the average value and \( \alpha_n \) are the Fourier modes \cite{Arfken2013}. 

To compute the Fourier modes \( \alpha_n \), we utilize the expression:
\begin{equation}
\label{fourier_modes}
    \alpha_n = 4\sum_{k, m = 1}^{d-1}\sum_{j > k}^d\sum_{l > m}^d|c_{j}|^2|c_{k}|^2|c_{l}|^2|c_{m}|^2\delta^{n}_{(j-k)(l-m)},
\end{equation}
where \( \delta^{n}_{(j-k)(l-m)} \) represents the Kronecker delta function ensuring the resonance condition for the Fourier modes.
This formulation allows us to decompose the reduced purity \( \gamma_{A}(t) \) into its constituent Fourier components, facilitating the identification of PNs based on their distinct Fourier signatures.

For prime \( n \), the trivial decomposition is \( (j-k) = n \) and \( (l-m) = 1 \), and vice versa. This results in a unique decomposition that corresponds to the expected behavior for prime numbers. 
However, if \( n \) is composite, it possesses non-trivial decompositions as well. To examine the impact of these decompositions on the Fourier modes expressed in Eq. (\ref{fourier_modes}), let us define the lower bound \( B_n \) as the value obtained from Eq. (\ref{fourier_modes}) using the trivial decomposition of \( n \geq 2\). Hence, we have:
\begin{equation}
\label{bound}
    B_n = 8\sum_{k = 1}^{d - n}\sum_{m = 1}^{d - 1}|c_{k}|^2|c_{m}|^2|c_{k + n}|^2|c_{m + 1}|^2.
\end{equation}
This lower bound \( B_n \) provides insight into the minimum value that the Fourier coefficient \( \alpha_n \) can attain for a given composite \( n \). Understanding this bound is crucial for discerning the distinct Fourier signatures associated with prime and composite numbers.
For \( 2 \leq n \leq d - 1 \), we have \( B_n > 0 \) as per Eq. (\ref{bound}). However, when \( d \leq n \leq (d-1)^2 \), the domain of \( B_n \) can be extended such that \( B_n = 0 \).

Now, let \( \{y^{(n)}_r\}_{r = 1}^{z} \) represent the sequence of \( z \) distinct divisors of \( n \geq 2 \) in increasing order of magnitude. Excluding the trivial cases \( y^{(n)}_1 = 1 \) and \( y^{(n)}_z = n \), we find that in general:
\begin{equation}
\label{alpha_n_full}
    \alpha_n = B_n + 4\sum_{r = 2}^{z - 1}\sum_{k = 1}^{d - \frac{n}{y^{(n)}_r}}\sum_{m = 1}^{d - y^{(n)}_r}|c_{k}|^2|c_{m}|^2|c_{k + \frac{n}{y^{(n)}_r}}|^2|c_{m + y^{(n)}_r}|^2.
\end{equation}
This expression for \( \alpha_n \) encompasses both the contribution from the trivial decomposition and the contributions from the non-trivial divisors of \( n \), enabling a comprehensive assessment of the Fourier modes associated with composite numbers.

In the domain \( 2 \leq n \leq 2(d-1) \), we can confidently assert that \( \alpha_n > B_n \) holds true. However, beyond this range, specifically in the interval \( 2(d-1) < n \leq (d-1)^2 \), certain composite numbers \( n = n_0 \) may exhibit \( \alpha_{n_0} = 0 \). This phenomenon arises because the first semi-primes (numbers that are the product of two prime numbers) are multiples of \( 2 \). Consequently, for \( n_0 = 2v \), where \( v > d - 1 \) is a prime, there exist no values for the indices \( k \) and \( m \) in Eq. (\ref{alpha_n_full}) that fall within their defined ranges in the summation. However, in this interval it is possible to discard any integer as a prime candidate if it has a non-zero Fourier mode. Since prime numbers always yield $\alpha_n = 0$ in this interval, we can safely guarantee that if $\alpha_n \neq 0$, then $n$ is composite. The inverse, however, is not always true: some composite numbers have $\alpha_n = 0$.



Here is the summary of the expected values of \( \alpha_n \) in the three regimes:

- Regime \textbf{I}: \(2\leq n \leq d - 1\). For prime numbers in this range, it holds true that \( \alpha_n = B_n > 0 \); otherwise, \( \alpha_n > B_n \). 

- Regime \textbf{II}: \(d \leq n \leq 2(d-1)\). Prime numbers in this interval exhibit \( \alpha_n = B_n = 0 \), while composite numbers consistently demonstrate \( \alpha_n > 0 \).

- Regime \textbf{III}: \(2(d-1) < n \leq (d-1)^2\). Prime numbers within this regime always yield \( \alpha_n = 0 \). However, some composite numbers may also yield \( \alpha_n = 0 \) in this interval. Consequently, this regime cannot provide conclusive evidence regarding the primality of \( n \). Nonetheless, any integer $n$ with $\alpha_n \neq 0$ in this regime is guaranteed to be a composite number.

This summary provides a clear delineation of the behavior of \( \alpha_n \) across different regimes, aiding in the identification of prime numbers based on their Fourier modes.

Our regime of interest is \( \mathcal{D} = \textbf{I} \cup \textbf{II} \). In \( \mathcal{D} \), it is consistently true that:
\begin{equation}
\label{prime_inequality}
\alpha_n \geq B_n,
\end{equation}
with equality achieved if and only if \( n \) is a prime number. This inequality forms the cornerstone of the algorithm and serves as the basis for objectively distinguishing prime numbers from composites.

While the protocol enables the computation of \( \alpha_n \), without knowledge of \( B_n \) in Regime \textbf{I}, it is impossible to discern whether \( \alpha_n = B_n \) or \( \alpha_n \neq B_n \). 
A straightforward solution involves obtaining the analytical value of the lower bound \( B_n \) within that regime, achievable by selecting a simple initial state and utilizing Eq. (\ref{bound}) subsequently.
In our algorithm, for simplicity, we opt for an initial state of maximum superposition.

\section{The Quantum algorithm}
\label{sec:algorithm}

Below, we provide a structured description of all the steps necessary to develop our protocol. We also present here the number of gates necessary for each step.

\textit{1. Qubit Codification}: To adapt our protocol to a qubit-based quantum computing algorithm, we need to adjust some of our definitions regarding the translation of qudits to qubits. Given that the bipartite system \(AB\) has \(d^2\) energy levels and we aim to utilize \(q\) qubits instead of two qudits, the condition is imposed that:
\begin{equation}
\label{codification}
        d^2 = 2^q.
\end{equation}
Equation (\ref{codification}) inherently assumes that \(d\) is a power of \(2\). If $d$ is not a power of $2$, we have to find $q$ such that \(q = 2\lceil \log_2 (d)\rceil\), where \(\lceil . \rceil\) denotes the ceiling function. We conveniently assign the first half of qubits to represent subsystem \(A\) and the remaining half to represent subsystem \(B\).

\textit{2. Initial State Flexibility}: The initial state \( |\psi(0)\rangle_{AB} \) is defined as the product state \( |\psi(0)\rangle_{AB} = |\phi\rangle_A \otimes |\phi\rangle_B \), where the coefficients \( c_{n_S} \) of the subsystem states \( |\phi\rangle_S = \sum_{n_S = 1}^d c_{n_S} |E_{n_S}\rangle \) must satisfy \( c_{n_S} \neq 0 \). Leveraging this degree of freedom, we opt for convenience by employing an initial state that achieves maximum superposition, expressed as \( |\psi(0)\rangle_{AB} = \frac{1}{d}\sum_{n_A, n_B = 1}^d |E_{n_A}E_{n_B}\rangle \). Here, we implicitly define the eigenbasis \( \{E_{n_S}\}_{n_S = 1}^d \) as the computational basis for each set of \( q/2 \) qubits. To produce this initial state, we apply a series of Hadamard gates \( \hat{H} \) \cite{Nielsen2000} to all \( q \) qubits:
\begin{equation}
\label{psi_0}
|\psi(0)\rangle_{AB} = \hat{H}^{\otimes q}|000...0\rangle.
\end{equation}
It is evident that the number of gates required here to generate this initial state is simply:
\begin{equation}
\label{G1}
G_1(q) = q.
\end{equation}

\textit{3. Evolved State Preparation}: The detailed results regarding this item are provided in Appendix \ref{sec:appendixA}. To obtain the evolved state \( |\psi(t)\rangle_{AB} \) of Eq. (\ref{psi_t}), we employ the method outlined in Ref. \cite{welch} to construct \( \hat{U}(t) \) efficiently. Initially, we have to determine, in principle, all the \( 2^q - 1 \) Walsh angles \( a_j(t) \) \cite{Walsh1923, Fine1949, Zhihua1983, Yuen1975}. However, according to the results shown in Appendix \ref{sec:appendixD}, only $\frac{1}{4}q^2 + q$ of them are non-null. By definition, Walsh angles are expressed as
\begin{equation} 
\label{walsh_angles}
a_j(t) = \frac{1}{2^q}\sum_{k = 0}^{2^q - 1}f_k(t)w_{jk},
\end{equation}
where \( w_{jk} \) denotes the Paley-ordered discrete Walsh functions and \( f_k(t) \) are the eigenvalues of the operator $\hat{f}(t)$, extracted from $\hat{U}(t) = e^{i\hat{f}(t)}$. 

Together with the Walsh angles \( a_j(t) \), the unitary gate \( \hat{U}(t) \) is obtained using the formalism of Walsh operators \( \hat{w}_j \). The expression for $\hat{U}(t)$ is then given by
\begin{equation}
\label{U_walsh}
    \hat{U}(t) = \prod_{j = 1}^{2^q - 1}e^{ia_j(t)\hat{w}_j}.
\end{equation}
To produce the exponential operators \(e^{ia_j(t) \hat{w_j}} \), we use the identity presented in Appendix \ref{sec:appendixB} and consider the binary representation \( (j_{q}...j_{2}j_{1}) \) of the integer \( j \), with the most significant non-zero bit (MSB) on the left. This enables us to represent the exponential operator as a single Z-rotation, $\hat{R}_z(\theta_j(t))$, applied to qubit $q_{m_{h_j}}$, flanked by two identical controlled-NOT gates, with qubit $q_{m_i}$ serving as the control and qubit $q_{m_{h_j}}$ as the target. Here, the index $m_{h_j} \geq 1$ signifies the position of the MSB of $j$, and the indices $m_i$ are defined by the condition $j_{m_i} = 1$. These $\hat{R}_z (\theta_j(t))$ rotations have angles \( \theta_j(t) = -2a_j(t) \).  

Therefore, preparing $|\psi(t)\rangle_{AB} = \hat{U}(t)|\psi(0)\rangle_{AB}$ demands a number of gates given by
\begin{equation}
\label{G2}
    G_2(q) = \frac{3}{4}q^2 + q. 
\end{equation}

\textit{4. Reduced Purity Estimation}: This step involves efficiently obtaining the reduced purity \( \gamma_A(t) \) of Eq. (\ref{purity}) by utilizing techniques from Ref. \cite{ekert}. The quantum circuit employed here resembles the SWAP test circuit \cite{buhrman, barenco} and employs an ancilla qubit \( q_0 \) and two copies of \( q \) qubits prepared in the same pure state. The operations sequence for this quantum circuit is as follows: a Hadamard gate on \( q_0 \), qubit-qubit controlled-SWAP gates between the first $q/2$ qubits of each copy, with $q_0$ as the control qubit, another Hadamard gate on \( q_0 \) and a measurement of $q_0$ in the computational basis. After repeatedly executing the circuit, we estimate the probability $P_0$ of obtaining the state $|0\rangle$ for $q_0$. Then, as detailed in Appendix \ref{sec:appendixE}, the reduced purity over time can be estimated using the expression \( \gamma_A(t) = 2P_0(t) - 1 \).

This step involves a total number of gates given by
\begin{equation}
\label{G3}
    G_3(q) = \frac{3}{2}q + 2.
\end{equation}

\textit{5. Fourier Modes Calculation}: In Regime \textbf{I}, we obtain the lower bound \( B_n \) of Eq. (\ref{bound}) using the initial state of Eq. (\ref{psi_0}). In this case, \( c_j = \frac{1}{\sqrt{d}} \) for any \( j \), and the corresponding lower bound \( B_n \) interpolation in this range of \( n \) is a straight line with a negative slope. In Regime \textbf{II}, the lower bound is \( B_n = 0 \). The expression for \( B_n \) in the regime of interest \( \mathcal{D} = \) \textbf{I} \( \cup \) \textbf{II} can then be written as 
\[
B_n = 
\begin{cases} 
\frac{-8(d-1)}{d^4}n + \frac{8d - 8}{d^3} & \text{if } n \in \textbf{I}, \\
0 & \text{if } n \in \textbf{II}.
\end{cases}
\]
Considering the remarks made in the previous section, we know that in a graph of Fourier modes, every prime number must have a corresponding \( \alpha_n \) position belonging exactly to the interpolated curve of \( B_n \). Any composite number in the regime of interest \( \mathcal{D} \) has \( \alpha_n > B_n \) and thus is necessarily above \( B_n \).

Now, using Fourier analysis, the Fourier modes \( \alpha_n \) are calculated by the integral
\begin{equation} \label{alpha_numerico}
\alpha_n = \frac{2\omega}{\pi}\int_{0}^{T/2}\gamma_A(t)\cos(n\omega t)dt.
\end{equation}
Normally, Eq. (\ref{alpha_numerico}) would be an integral over the whole period \( T \), but we are employing the property of the symmetry of \( \gamma_A(t) \), presented earlier in this article. After calculating the Fourier modes \( \alpha_n \), the last part of our algorithm involves comparing the value of \( \alpha_n \) with the analytical lower bound \( B_n \). In this final step, the numerical integration is done in \( p \) partitions, resulting in an equivalent number of points used for \( \gamma_A(t) \) in the interval \( 0 \leq t \leq T/2 \). Consequently, to achieve a desired precision \( \epsilon \), our quantum circuit requires at least \( p \) executions. Currently, the exact optimal scaling of \( p \) with respect to \( d \), for a given \( \epsilon \), remains undetermined.

\begin{figure}[t]
  \centering    \includegraphics[width=0.80\linewidth]{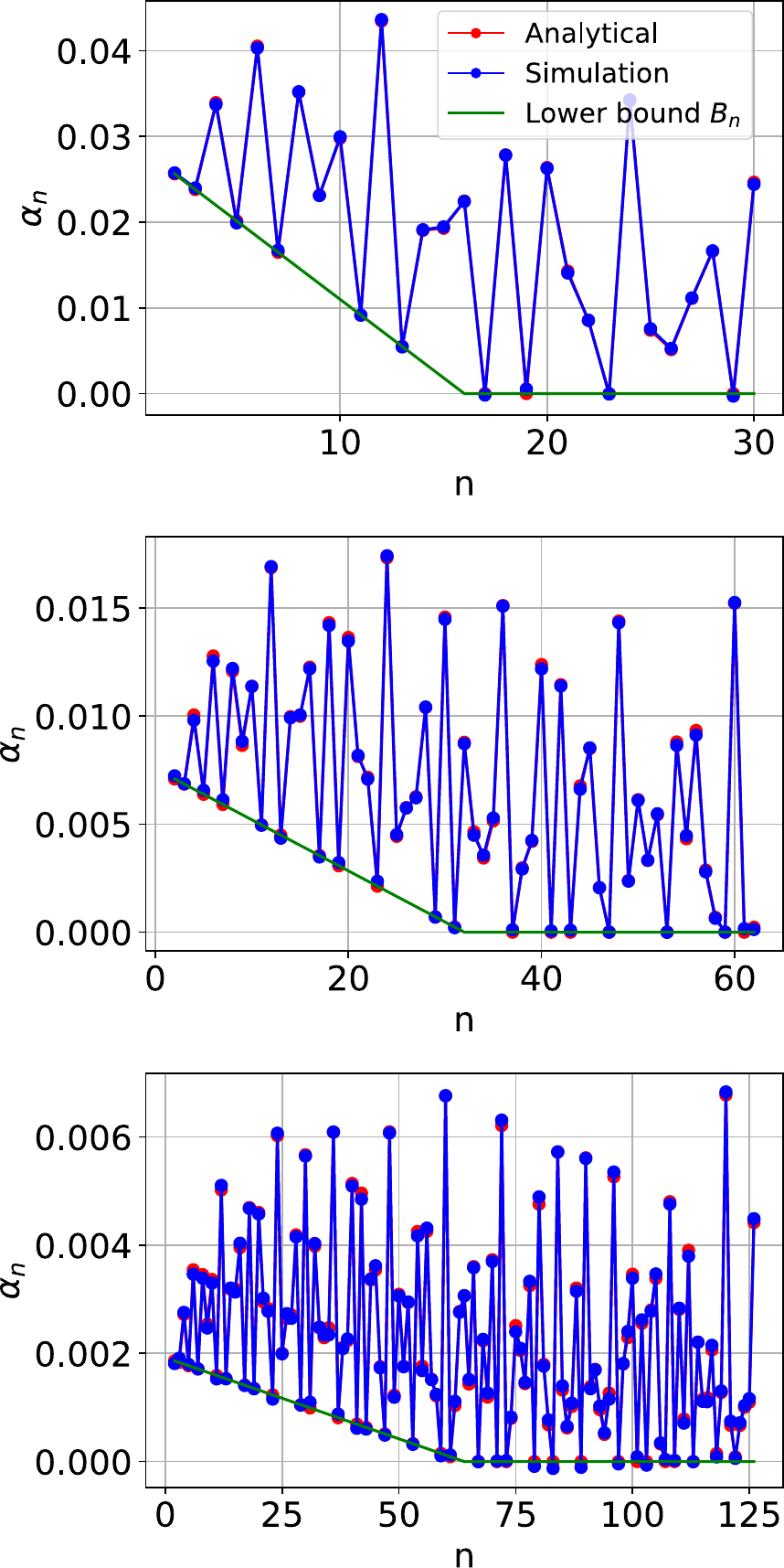}  \captionsetup{justification=raggedright,singlelinecheck=false}
\caption{Comparison of simulation results with theoretical predictions for the Fourier modes of the reduced purity across different dimensions $(d = 16, 32, 64)$. Red points represent the analytical values of $\alpha_n$, calculated directly from their theoretical expressions, while the green line stands for the minimum value for the Fourier modes, also derived from theoretical calculations. Blue points illustrate the Fourier modes obtained through numerical integration of the reduced purity $\gamma_A(t)$ extracted from the classical emulation of our quantum circuit. The numerical integration was performed using various partition values $(p = 375, 1500, 6000)$. Prime numbers are expected to align with the lower bound $B_n$, whereas composite numbers appear above.}
    \label{fig:fourier}
\end{figure}

\section{Simulations}
\label{sec:simulations}

In order to evaluate the applicability of our algorithm, classical simulations were performed using IBM's Qiskit framework (version 0.45.1). These simulations targeted three distinct values of \( d \), with results depicted in blue in Fig. \ref{fig:fourier}. For all the simulations, we used \( 10^5 \) shots and fixed \( \omega = 0.1 \, \text{s}^{-1} \), as changing the value of \( \omega \) has no effect on the Fourier modes \( \alpha_n \).
Regarding the number of executions \( p \) of the circuit, we selected \( p = 375 \), \( p = 1500 \), and \( p = 6000 \) for the dimensions \( d = 16 \), \( d = 32 \), and \( d = 64 \), respectively. The values of \( p \) were chosen to achieve roughly the same accuracy for the three values of \( d \). Using Python (version 3.11.3) with the Scipy library (version 1.11.3), Fourier modes \( \alpha_n \) were calculated with Simpson's rule for the numerical integration of Eq. (\ref{alpha_numerico}). Due to the substantial size of the quantum circuit for the three dimensions analyzed in our simulations, we present the circuit for a lower dimension, \( d=4 \), purely for illustrative purposes. This simplified example is shown in Fig. \ref{fig:qc}, allowing us to convey the structure without the complexity of the larger dimensions.

\begin{figure*}[t]
  \centering     \includegraphics[width=0.95\linewidth]{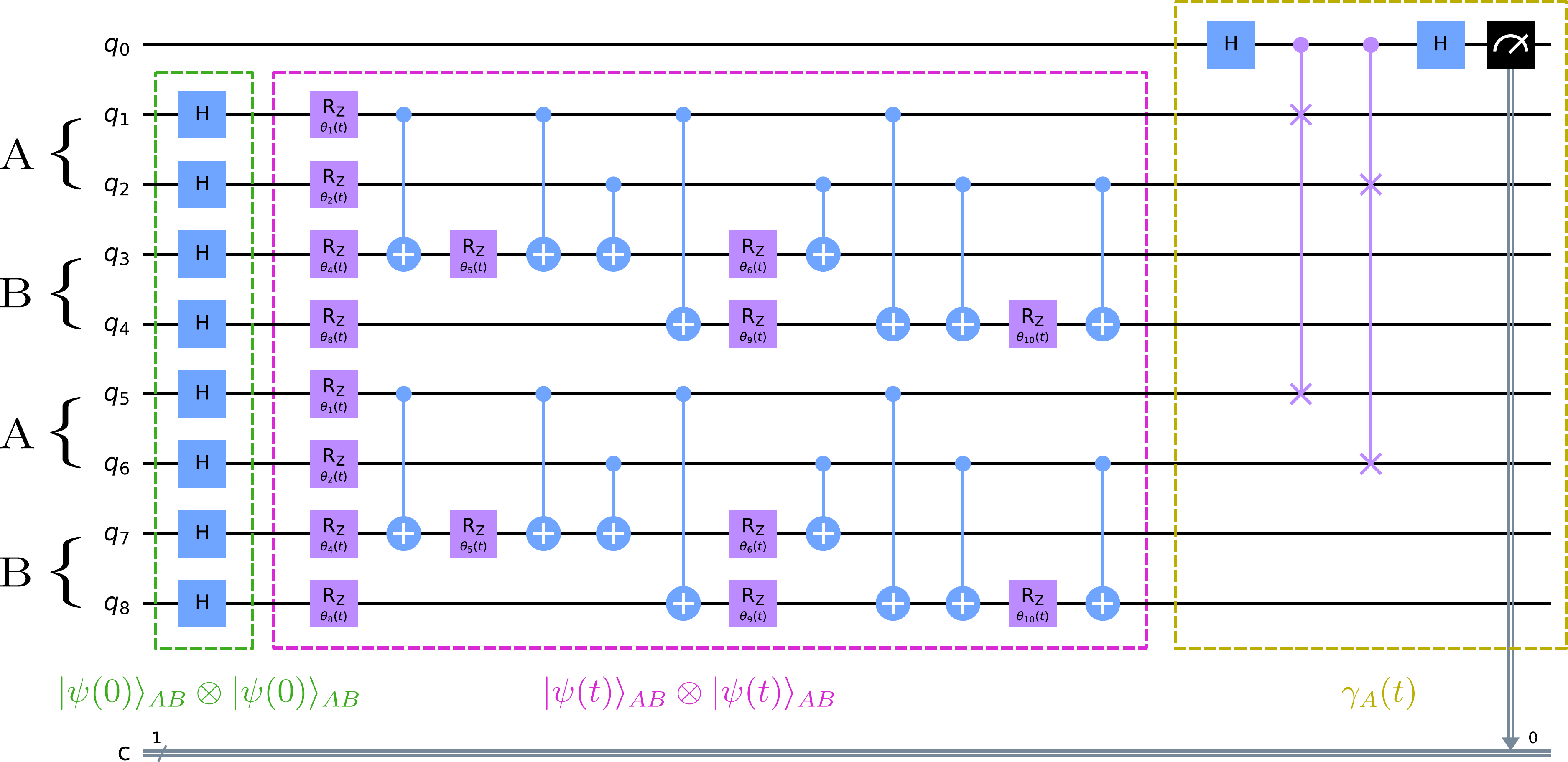}    \captionsetup{justification=raggedright,singlelinecheck=false}
    \caption{Schematic of the quantum circuit for our algorithm for \(d = 4\). Each colored box corresponds to a stage of the circuit and represents the set of operations in the respective step for the two copies of \(q\) qubits. In the first stage, we prepare the initial state of maximum superposition \(|\psi(0)\rangle_{AB}\) for both copies of \(q = 4\) qubits, starting from the state $|0\rangle$ for each qubit. The second stage is used for the efficient state preparation of \(|\psi(t)\rangle_{AB}\) for each copy, where the total number of gates, including both copies, is given by \(\frac{3}{2}q^2 + 2q\). In the last stage, with the aid of an ancilla qubit \(q_0\), we apply the gates corresponding to the variation of the SWAP test to extract the reduced purity \(\gamma_A(t)\) after executing the circuit several times.}
    \label{fig:qc}
\end{figure*}

\section{Conclusions} 
\label{sec:conclusions}

Concluding, this work presented a qubit-based quantum algorithm for prime number identification, rooted in the analysis of entangled subsystem dynamics. By employing a bipartite Hamiltonian and analyzing the Fourier modes of the reduced purity, we distinguish between prime and composite numbers within the range $2 \leq n \leq 2(d-1)$. Implementing this on a qubit-based system involves transforming a 2-qudit system into a qubit system, with the unitary gate of Eq. (\ref{unitary}) implementable in polynomial time, contrary to the expected exponential gate requirements.

Our quantum circuit executes in three stages with a total gate cost indicating quadratic scaling in the number of digits of \(N = 2(d-1)\). Despite idealized simulations, implementation on quantum hardware is feasible but faces challenges such as qubit connectivity. Alternatives like trapped ion quantum computers or modified gate preparation and reduced purity measurement methods could overcome these.

The efficient realization of unitary operations demonstrates the potential for broader application in quantum computing, suggesting future work could extend this algorithm to verify larger primes. This progress in quantum algorithm optimization could significantly impact the field's practical application to fundamental computational problems.

\begin{acknowledgments}
This work was supported by the Coordination for the
Improvement of Higher Education Personnel (CAPES),
Grant No. 23081.031640/2023-17, by the National Council for Scientific and Technological Development (CNPq), Grants No. 309862/2021-3, No. 409673/2022-6, and No. 421792/2022-1, and the National   Institute for the Science and Technology of Quantum Information (INCT-IQ), Grant No. 465469/2014-0. We thank Alexandre D. Ribeiro for valuable discussions on the subject of this article.
\end{acknowledgments}

\textbf{Data availability.} The data that support the findings of this study are available at \textcolor{blue}{https://github.com/santosvictorf/primes-identification-using-qcomputers/tree/main/qiskit}. This repository includes the Python code for implementing the quantum algorithm in Qiskit, the simulation results, and auxiliary codes that support the theoretical findings.

\appendix

\section{Diagonal unitary gate implementation using Walsh functions}
\label{sec:appendixA}

In this Appendix, we provide an overview of the algorithm introduced in Ref. \cite{welch} for implementing unitary operations on quantum computers. To begin, we establish some definitions. Let $q$ denote the number of qubits, and consider positive integers $j$ and $k$ with binary and dyadic representations given by
\begin{align}
    & \text{bin}(j) = (j_qj_{q-1}\cdots j_1), \\
    & \text{dyad}(k) = (k_1k_2\cdots k_q),
\end{align}
where the most significant bit (MSB) is on the left. Henceforth, we assume $j = 0, 1, \cdots, 2^q - 1$ and $k = 0, 1, \cdots, 2^q - 1$.

Next, we define the discrete Paley-ordered Walsh functions $w_{jk}$ as
\begin{equation}
\label{walsh_functions}
    w_{jk} = (-1)^{\sum_{i = 1}^q j_ik_i}.
\end{equation}
Let us discretize the interval $0 \leq x < 1$ into points given by
\begin{equation}
    x_k = \frac{k}{2^q}.
\end{equation}
Since the Walsh functions form an orthonormal basis, we can define the Walsh-Fourier transform for a function $f_k = f(x_k)$ as follows:
\begin{align}
\label{walsh_angle}
    a_j &= \frac{1}{2^q}\sum_{k = 0}^{2^q - 1} f_k w_{jk}, \\
    f_k &= \sum_{j = 0}^{2^q - 1} a_j w_{jk}.
\end{align}

In qubit-based quantum computing, the state of $q$ qubits generally takes the form $|\psi\rangle = \sum_{k = 0}^{2^q - 1} c_k |k\rangle$, where the computational basis $|k\rangle$ is defined as
\begin{equation}
    |k\rangle = |k_1k_2\cdots k_q\rangle,
\end{equation}
with $k$ represented in dyadic form $\operatorname{dyad(k)} = (k_1k_2\cdots k_q)$. Now, let us define the unitary operator $\hat{U} = e^{i\hat{f}}$ \cite{Nielsen2000}, where $\hat{f}$ is a diagonal operator in the computational basis:
\begin{equation}
    \hat{f} |k\rangle = f_k |k\rangle.
\end{equation}

Walsh operators $\{\hat{w}_j\}_{j = 0}^{2^q - 1}$ acting on $q$ qubits are naturally defined as
\begin{equation}
\label{walsh_operators}
    \hat{w}_j = (\hat{Z_1})^{j_1} \otimes (\hat{Z_2})^{j_2} \otimes \cdots \otimes (\hat{Z_q})^{j_q},
\end{equation}
where $(\hat{Z}_i)^1 = \hat{Z}_i$ represents the Pauli $Z$ operator and $(\hat{Z}_i)^0 = \hat{I}$ denotes the identity matrix, both acting on the $i$-th qubit $q_i$. This definition of Walsh operators is advantageous because their action on the computational basis is given by
\begin{equation}
    \hat{w}_j|k\rangle = w_{jk}|k\rangle.
\end{equation}

This implies that the eigenvalues of Walsh operators $\hat{w}_j$ are the Walsh functions $w_{jk}$, and these operators form a basis for diagonal operators $\hat{f}$. Additionally, due to their form, Walsh operators commute. Therefore, considering $\hat{w}_0 = \hat{I}$, we can disregard $j = 0$, leading to the expression
\begin{equation}
\label{walsh_unitary}
    \hat{U} = e^{i\hat{f}} = \prod_{j = 1}^{2^q - 1} e^{ia_j\hat{w}_j}.
\end{equation}

In essence, to apply the method outlined in Ref. \cite{welch}, we begin by determining the $f_k$ values associated with the unitary gate $\hat{U}$. Subsequently, we construct the Walsh functions $w_{jk}$ using the procedure described in Appendix \ref{sec:appendixC}. With these components in hand, Eq. (\ref{walsh_angle}) allows us to compute the Walsh angles $a_j$. Finally, utilizing the identity presented in Appendix \ref{sec:appendixB} to construct the $\hat{w}_j$ operators in Eq. (\ref{walsh_unitary}) yields the desired unitary $\hat{U}$ with a gate cost of $\mathcal{O}(2^q)$ in general. This gate cost can be optimized by reordering the commuting exponential operators in Eq. (\ref{walsh_unitary}) using the Gray code. It is important to note that even with optimal construction, the quantum circuit for this method typically requires $\mathcal{O}(2^q)$ gates. However, as we will demonstrate in Appendix \ref{sec:appendixD}, for the specific case of the $q$-qubit unitary gate $\hat{U}(t)$ described in Eq. (\ref{unitary}), implementation with a polynomial gate cost is achievable by identifying the null Walsh angles $a_j(t)$.

\section{Relation between Pauli Z gates and CNOTs staircases}
\label{sec:appendixB}

In this Appendix, we delve into a fundamental identity pivotal to our analysis, which concerns the tensor product of Pauli $\hat{Z}$ operators. This identity plays a crucial role in simplifying the representation of quantum states and operations within our framework. To lay the groundwork for our discussion, we introduce essential notation and concepts: \\
- \( h_j \), the Hamming weight of \( j \), represents the number of \( 1 \)'s in the binary representation of \( j \), corresponding to the number of \( \hat{Z} \) operators in the tensor product. \\
- The identity operator \( \hat{I}^{\otimes r} \) acts on \( r \) qubits, serving as a placeholder in tensor products where no operation is performed. \\
- The operators \( \hat{A}_{h_j} \) are constructed from a sequence of controlled-NOT ( \( \widehat{CNOT} \) ) gates, defined as \( \hat{A}_{h_j} = \widehat{CNOT}_{h_j}^{1}  \widehat{CNOT}_{h_j}^{2} \cdots \widehat{CNOT}_{h_j}^{h_j - 1} \), where \( \widehat{CNOT}_{b}^{a} \) denotes a \( \widehat{CNOT} \) gate with qubit \( q_a \) as the control and qubit \( q_b \) as the target.

With these definitions in place, we establish the following identity:
\begin{equation} \label{ZZZ}
\hat{Z}_1 \otimes \hat{Z}_2 \otimes \cdots \otimes \hat{Z}_{h_j - 1} \otimes \hat{Z}_{h_j} = \hat{A}_{h_j}\left(\hat{I}^{\otimes (h_j - 1)} \otimes \hat{Z}\right) \hat{A}_{h_j}^{-1}.
\end{equation}
This identity demonstrates how a tensor product of \( \hat{Z} \) operators can be equivalently expressed through a transformation involving \( \hat{A}_{h_j} \) and its inverse, significantly simplifying the representation and manipulation of such operations. Building upon this foundation, we further examine its implications in the exponential form:
\begin{equation} \label{exp_Z}
e^{ia_j \left(\hat{Z}_1 \otimes \hat{Z}_2 \otimes \cdots \otimes \hat{Z}_{h_j - 1} \otimes \hat{Z}_{h_j}\right)} = \hat{A}_{h_j}\left(\hat{I}^{\otimes (h_j - 1)} \otimes e^{ia_j\hat{Z}}\right) \hat{A}_{h_j}^{-1}.
\end{equation}
This expression further underscores the utility of the \( \hat{A}_{h_j} \) transformation in facilitating the implementation of the exponential quantum gates \( e^{ia_j\hat{w}_j} \).

Now, we proceed with the proofs. For the calculations below, unless otherwise convenient, we do not specify the qubit index $i$ of Pauli operators $\hat{Z}_i$ or any other operators. We start by rewriting the left side of Eq. (\ref{ZZZ}) using the projectors $\hat{\Pi}_0 = |0\rangle\langle 0|$ and $\hat{\Pi}_1 = |1\rangle\langle 1|$:
\begin{align}
& \hat{Z} \otimes \hat{Z} \otimes \cdots \otimes \hat{Z} \otimes \hat{Z} \nonumber \\
&= (\hat{\Pi}_0 - \hat{\Pi}_1) \otimes (\hat{\Pi}_0 - \hat{\Pi}_1) \otimes \cdots \otimes (\hat{\Pi}_0 - \hat{\Pi}_1) \otimes \hat{Z}\notag \\
&= \sum_{\operatorname{bin(s)}} (-1)^{h_s} \hat{\Pi}_{s_1} \otimes \hat{\Pi}_{s_2} \otimes \cdots \otimes \hat{\Pi}_{s_{(h_j - 1)}} \otimes \hat{Z} \notag \\
\label{sum_proj}
&= \sum_{\substack{\operatorname{bin(s)} \\ 
h_s \text{ even}}} \hat{\Pi}_{s_1} \otimes \hat{\Pi}_{s_2} \otimes \cdots \otimes \hat{\Pi}_{s_{(h_j - 1)}} \otimes \hat{Z} \nonumber \\
& - \sum_{\substack{\operatorname{bin(s)} \\ 
h_s \text{ odd}}} \hat{\Pi}_{s_1} \otimes \hat{\Pi}_{s_2} \otimes \cdots \otimes \hat{\Pi}_{s_{(h_j - 1)}} \otimes \hat{Z},
\end{align}
where the sum on $\operatorname{bin(s)}$ concerns all the possible binary representations $\operatorname{bin(s)}= (s_{(h_j - 1)}s_{(h_j-2)}\cdots s_1)$ of $h_j - 1$ bits. It will be helpful to define 
\begin{align}
\Sigma_{\text{even}} &= \sum_{\substack{\operatorname{bin(s)} \\ h_s \text{ even}}} \hat{\Pi}_{s_1} \otimes \hat{\Pi}_{s_2} \otimes \cdots \otimes \hat{\Pi}_{s_{(h_j - 1)}}, \\
\Sigma_{\text{odd}} &= \sum_{\substack{\operatorname{bin(s)} \\ h_s \text{ odd}}} \hat{\Pi}_{s_1} \otimes \hat{\Pi}_{s_2} \otimes \cdots \otimes \hat{\Pi}_{s_{(h_j - 1)}}.
\end{align}
For these two definitions, the following relations are inherited from the projectors:
\begin{align}
& \Sigma_{\text{even}} \Sigma_{\text{even}} = \Sigma_{\text{even}}, \\
& \Sigma_{\text{odd}} \Sigma_{\text{odd}} = \Sigma_{\text{odd}}, \\
& \Sigma_{\text{even}} \Sigma_{\text{odd}} = \Sigma_{\text{odd}} \Sigma_{\text{even}} = 0.
\end{align}

Therefore, after defining $\hat{I}$ as the identity gate acting on a single qubit and recalling that $\hat{X}\hat{Z}\hat{X} = -\hat{Z}$, we obtain:
\begin{align}
& \hat{Z} \otimes \hat{Z} \otimes \cdots \otimes \hat{Z} \otimes \hat{Z} \nonumber \\
&= \Sigma_{\text{even}} \otimes \hat{Z} - \Sigma_{\text{odd}} \otimes \hat{Z}\notag \\
&= \Sigma_{\text{even}} \otimes \hat{Z} + \Sigma_{\text{odd}} \otimes \hat{X}\hat{Z}\hat{X} \notag \\  
&= \left(\Sigma_{\text{even}} \otimes \hat{I}\right)\left(\hat{I}^{\otimes (h_j - 1)} \otimes \hat{Z} \right)\left(\Sigma_{\text{even}} \otimes \hat{I}\right) \nonumber \\
& + \left(\Sigma_{\text{odd}} \otimes \hat{X}\right)\left(\hat{I}^{\otimes (h_j - 1)} \otimes \hat{Z} \right)\left(\Sigma_{\text{odd}} \otimes \hat{X} \right) \notag \\
&= \left(\Sigma_{\text{even}} \otimes \hat{I}\right)\left(\hat{I}^{\otimes (h_j - 1)} \otimes \hat{Z} \right)\left(\Sigma_{\text{even}} \otimes \hat{I}\right) \nonumber \\
& + \left(\Sigma_{\text{odd}} \otimes \hat{X}\right)\left(\hat{I}^{\otimes (h_j - 1)} \otimes \hat{Z} \right)\left(\Sigma_{\text{odd}} \otimes \hat{X} \right) \notag \\
&\quad + \left(\Sigma_{\text{odd}} \otimes \hat{X}\right)\left(\hat{I}^{\otimes (h_j - 1)} \otimes \hat{Z} \right)\left(\Sigma_{\text{even}} \otimes \hat{I}\right) \nonumber \\
& + \left(\Sigma_{\text{even}} \otimes \hat{I}\right)\left(\hat{I}^{\otimes (h_j - 1)} \otimes \hat{Z} \right)\left(\Sigma_{\text{odd}} \otimes \hat{X} \right) \notag \\
\label{Z_even_odd}
&=  \left(\Sigma_{\text{even}} \otimes \hat{I} + \Sigma_{\text{odd}} \otimes \hat{X}\right)\left(\hat{I}^{\otimes (h_j - 1)} \otimes \hat{Z} \right) \nonumber \\
& \times\left(\Sigma_{\text{even}} \otimes \hat{I} + \Sigma_{\text{odd}} \otimes \hat{X}\right).
\end{align}
To continue, we examine the product of two controlled-NOT gates targeting the same qubit:
\begin{align}    & \widehat{CNOT}_{3}^{1}\widehat{CNOT}_{3}^{2} \nonumber \\
&= \left(\hat{\Pi}_0 \otimes \hat{I} \otimes \hat{I} + \hat{\Pi}_1 \otimes \hat{I} \otimes \hat{X}\right) \nonumber \\ 
&\times \left(\hat{I} \otimes \hat{\Pi}_0 \otimes \hat{I} + \hat{I} \otimes \hat{\Pi}_1 \otimes \hat{X}\right) \notag \\
    &= \hat{\Pi}_0 \otimes \hat{\Pi}_0 \otimes \hat{I} + \hat{\Pi}_0 \otimes \hat{\Pi}_1 \otimes \hat{X} \nonumber \\ 
    & + \hat{\Pi}_1 \otimes \hat{\Pi}_0 \otimes \hat{X} + \hat{\Pi}_1 \otimes \hat{\Pi}_1 \otimes \hat{I} \notag \\
    \label{prod_cnot}
    &= \left(\hat{\Pi}_0 \otimes \hat{\Pi}_0 + \hat{\Pi}_1 \otimes \hat{\Pi}_1\right) \otimes \hat{I} \nonumber \\ 
    & + \left(\hat{\Pi}_1 \otimes \hat{\Pi}_0 + \hat{\Pi}_0 \otimes \hat{\Pi}_1\right) \otimes \hat{X}. \notag \\
\end{align}
The equation above suggests a similar form for a more general case. In fact, it holds that
\begin{align}
    \hat{A}_{h_j} &= \widehat{CNOT}_{h_j}^{1}  \widehat{CNOT}_{h_j}^{2}\cdots \widehat{CNOT}_{h_j}^{h_j-1} \nonumber \\
    &= \sum_{\substack{\operatorname{bin(s)} \\ h_s \text{ even}}} \hat{\Pi}_{s_1} \otimes \hat{\Pi}_{s_2} \otimes \cdots \otimes \hat{\Pi}_{s_{(h_j - 1)}} \otimes \hat{I} \nonumber \\
    & + \sum_{\substack{\operatorname{bin(s)} \\ h_s \text{ odd}}} \hat{\Pi}_{s_1} \otimes \hat{\Pi}_{s_2} \otimes \cdots \otimes \hat{\Pi}_{s_{(h_j - 1)}} \otimes \hat{X} \notag \\
    \label{cnot_even_odd}
    &= \Sigma_{\text{even}} \otimes \hat{I} + \Sigma_{\text{odd}} \otimes \hat{X}.\
\end{align}

Then, because $\hat{A}_{h_j}^{-1} = \hat{A}_{h_j}$, we obtain the proposed expression (\ref{ZZZ}) by using the identity (\ref{cnot_even_odd}) on Eq. (\ref{Z_even_odd}):
\begin{equation}
\hat{Z}_1 \otimes \hat{Z}_2 \otimes \cdots \otimes \hat{Z}_{h_j - 1}\otimes \hat{Z}_{h_j}
 = \hat{A}_{h_j}\left(\hat{I}^{\otimes (h_j - 1)} \otimes \hat{Z}\right) \hat{A}_{h_j}^{-1}.    
\end{equation}
Using this result, we can further demonstrate the validity of Eq. (\ref{exp_Z}):
\begin{align}
    & e^{ia_j \left(\hat{Z}_1 \otimes \hat{Z}_2 \otimes \cdots \otimes \hat{Z}_{h_j - 1} \otimes \hat{Z}_{h_j}\right)} \nonumber \\
    &= \sum_{n = 0}^{\infty}\frac{\left(ia_j\right)^n}{n!}\left(\hat{Z} \otimes \hat{Z} \otimes \cdots \otimes \hat{Z}\right)^n\notag \\
    &= \sum_{n = 0}^{\infty}\frac{\left(ia_j\right)^n}{n!}\left(\hat{A}_{h_j}\left(\hat{I}^{\otimes (h_j - 1)} \otimes \hat{Z}\right) \hat{A}_{h_j}^{-1}\right)^n \notag \\
    &= \sum_{n = 0}^{\infty}\frac{\left(ia_j\right)^n}{n!}\hat{A}_{h_j}\left(\hat{I}^{\otimes (h_j - 1)} \otimes \hat{Z}^n\right) \hat{A}_{h_j}^{-1} \notag \\
    &= \hat{A}_{h_j}\left(\hat{I}^{\otimes (h_j - 1)} \otimes \sum_{n = 0}^{\infty}\frac{\left(ia_j \hat{Z}\right)^n}{n!}\right) \hat{A}_{h_j}^{-1} \notag \\
    \label{exp_Z_2}
    &= \hat{A}_{h_j}\left(\hat{I}^{\otimes (h_j - 1)} \otimes e^{ia_j\hat{Z}}\right) \hat{A}_{h_j}^{-1}.
\end{align}

In this context, we revisit the formulation of Walsh operators, as delineated in Eq. (\ref{walsh_operators}), represented by \(\hat{w}_j = (\hat{Z_1})^{j_1} \otimes (\hat{Z_2})^{j_2} \otimes \cdots \otimes (\hat{Z_q})^{j_q}\), where the action of \(e^{ia_j\hat{w}_j}\) on a \(q\)-qubit basis state \(|k\rangle = |k_1k_2\cdots k_q\rangle\) is considered. To elaborate on the analysis, we introduce a strategic reordering of the indices \(j_i\), segregating them into two distinct sets: the first, denoted by \(\{m_i\}_{i = 1}^{h_j}\), corresponds to indices where \(j_{m_i} = 1\), spanning the initial \(h_j\) bits; the latter set, \(\{m_i\}_{i = h_j + 1}^{q}\), encompasses indices with \(j_{m_i} = 0\), accounting for the remaining \(q - h_j\) bits. The accordingly reconfigured states of qubits \(q_i\) and the operators \((\hat{Z}_i)^{j_i}\) can be achieved by applying SWAP gates. This culminates in the revised Walsh operator \(\hat{\Omega}_j\) and revised basis state \(|k'\rangle\), articulated as:
\begin{align} \hat{\Omega}_j =& \left(\hat{Z}_{m_1} \otimes \hat{Z}_{m_2} \otimes \cdots \otimes \hat{Z}_{m_{h_j}}\right) \nonumber \\ 
& \otimes \left(\hat{I}_{m_{h_j + 1}} \otimes \hat{I}_{m_{h_j + 2}} \otimes \cdots \otimes \hat{I}_{m_q}\right), \\ 
|k'\rangle =& |k_{m_1}k_{m_2}\cdots k_{m_{q}}\rangle.\end{align}

Leveraging Eq. (\ref{exp_Z_2}), we have
\begin{widetext}
\begin{align}
    e^{ia_j \hat{\Omega}_j} |k'\rangle &= e^{ia_j\left(\hat{Z}_{m_1} \otimes \hat{Z}_{m_2} \otimes \cdots \otimes \hat{Z}_{m_{h_j}}\right) \otimes \left(\hat{I}_{m_{h_j + 1}} \otimes \hat{I}_{m_{h_j + 2}} \otimes \cdots \otimes \hat{I}_{m_q}\right)}|k_{m_1}k_{m_2}\cdots k_{m_q}\rangle\notag \\
    &= e^{ia_j\left(\hat{Z}_{m_1} \otimes \hat{Z}_{m_2} \otimes \cdots \otimes \hat{Z}_{m_{h_j}}\right)} \otimes \left(\hat{I}_{m_{h_j + 1}} \otimes \hat{I}_{m_{h_j + 2}} \otimes \cdots \otimes \hat{I}_{m_q}\right) |k_{m_1}k_{m_2}\cdots k_{m_q}\rangle\notag \\
    \label{exp_Omega}
    &= \left[\hat{A}_{h_j}\left(\hat{I}^{\otimes (h_j - 1)} \otimes e^{ia_j\hat{Z}}\right) \hat{A}_{h_j}^{-1} \right] \otimes \hat{I}^{\otimes (q - h_j)} |k_{m_1}k_{m_2}\cdots k_{m_q}\rangle.
\end{align}
\end{widetext}
As it was stated previously, our objective lies in the action of the original operator $e^{ia_j\hat{w}_j}$ on the original basis state $|k\rangle$, i.e., $e^{ia_j\hat{w}_j}|k\rangle$. However, through the application of the same SWAP gates used before to Eq. (\ref{exp_Omega}), we restore the original sequence of qubits and Pauli operators, thereby preserving the structural integrity of the action of the operator $e^{ia_j \hat{w}_j}$ on the basis state $|k\rangle$.

\section{Walsh matrices construction}
\label{sec:appendixC}

A well-established result in the literature states that any discrete Walsh function \(w_{jk}\) can be represented as a product of Rademacher functions, with the exception of \(w_{0k}\), which trivially remains a constant function \(w_{0k} = 1\). Rademacher functions, denoted as \(w_{l_rk}\), are Walsh functions where \(l_r\) is a power of 2, specifically \(l_r = 2^{(r - 1)}\).

To illustrate, consider a Walsh matrix \(w_{(.,.)}\), where \(w_{jk} = w_{(j,k)}\), representing Paley-ordered Walsh functions arranged in a square matrix of dimensions \(2^q \times 2^q\). In matrix formalism, any row \(w_{(j,.)}\) can be understood as the column-wise product of the corresponding \(h_j\) Rademacher rows \(\{w_{({l_{m_i}},.)}\}_{i = 1}^{h_j}\), organized by ascending order of magnitude \(l_{m_1} < l_{m_2} < \cdots < l_{m_{h_j}}\). Here, \(h_j\) signifies the Hamming weight of \(j\). The Rademacher rows \(w_{({l_{m_i}},.)}\) satisfy \(j = \sum_{i = 1}^{h_j}l_{m_i} = \sum_{i = 1}^{h_j}2^{(m_i - 1)}\).

Here, we introduce this result as a proposition with minor adjustments. This adaptation enables us to introduce the required notation for the theorem detailed in the next Appendix, which concerns the implementation of $\hat{U}(t)$. Consistent with convention, we index the rows and columns of our Walsh matrix starting from $0$, with the maximum index value being $2^q - 1$.

\begin{proposition}\label{proposition1} Consider a positive integer $r$ satisfying $1 \leq r \leq q$, where $l_r = 2^{(r-1)}$. Let $\{m_i\}_{i = 1}^{h_j}$ be a sequence of integers in increasing order of magnitude and let $j$ be any row index of the $2^q \times 2^q$ Walsh matrix $w_{(.,.)}$. Then, the row $w_{(j,.)}$ will satisfy only one of the following statements: 
\begin{enumerate}
    \item If $j = l_{r}$, then the Rademacher row $w_{(j,.)} = w_{(l_{r},.)}$ of the Walsh matrix is given by
\begin{equation}
    w_{(l_{r},.)} = \bigl[(R_{r})(-R_{r})(R_{r})(-R_{r})\cdots(R_{r})(-R_{r})\bigr],
\end{equation}
where $R_{r} = (+1)(\times T_{r})$ is a row of length $T_{r} = 2^{(q - r)}$ and $T_r$ is called the period of the row $w_{(l_{r},.)}$. The notation means that $w_{(l_{r},.)}$ is composed of sign alternating sequences of $T_{r}$ elements. These sequences are $R_{r}$ and $-R_{r}$ and all their elements are, respectively, equal to $+1$ and $-1$. 
    \item If $j \neq l_r$, we define $j = \sum_{i = 1}^{h_j}l_{m_i} = \sum_{i = 1}^{h_j}2^{(m_i - 1)}$, then the row $w_{(j,.)}$ of the Walsh matrix is given by
\begin{equation}
    w_{(j,.)} = \bigl[(R_{m_1})(-R_{m_1})(R_{m_1})(-R_{m_1})\cdots(R_{m_1})(-R_{m_1})\bigr],
\end{equation}
where $R_{m_1}$ is obtained from the recursive relation
\begin{equation}
    R_{m_{(i-1)}} = \bigl[(R_{m_i})(-R_{m_i})(R_{m_i})(-R_{m_i})\cdots(R_{m_i})(-R_{m_i})\bigr],
\end{equation}
with $R_{m_{0}} = w_{(j,.)}$ and $R_{m_{h_j}} = (+1)(\times T_{m_{h_j}})$. Each $R_{m_{i}}$, for $1 \leq i \leq h_j$, is a row of length $T_{m_{i}} = 2^{(q - m_i)}$.
\end{enumerate}
\end{proposition}

\begin{proof}
We shall prove each case of Proposition \ref{proposition1} separately.

\textbf{1. Case \(j = l_r\).}
The binary representation of $j = l_r = 2^{(r - 1)}$ is given by
\begin{equation}     
\operatorname{bin(j)} = (000\cdots01_r0\cdots000),
\end{equation}
where we have introduced the notation $1_{r} = j_r = 1$ to make it clear that the only non-zero binary element of $\operatorname{bin(j)}$ is in the position $r$. 
\begin{definition}
We define the exponent $S(j,k)$ of Eq. $(\ref{walsh_functions})$ as
\begin{equation}
   S(j,k) = \sum_{i = 1}^qj_ik_i. 
\end{equation}
If $S(j,k)$ is even, then $w_{(j,k)} = 1$. If $S(j,k)$ is odd, then $w_{(j,k)} = -1$.
\end{definition}
Then, for any $0 \leq k \leq 2^q - 1$, we have
\begin{align}
k_{r} &= 0 \implies S(j,k) = j_{r}k_{r} = \operatorname{even}, \\
k_{r} &= 1 \implies S(j,k) = j_{r}k_{r} = \operatorname{odd.}
\end{align}
Now, we consider only a partial dyadic representation string $\operatorname{dyad_{r}(k)}$ of $k$. Since the dyadic string of $k$ is the same as its reverse binary string, a partial dyadic string up to the position $r$ is defined as 
\begin{equation}
\label{dyad_partial}
 \operatorname{dyad_{r}(k)} = (k_{r}k_{r + 1}\cdots k_{(q - 1)}k_q).  
\end{equation}
There are $q - r + 1$ entries in the partial dyadic string of Eq. (\ref{dyad_partial}), resulting in a total of $2^{(q - r + 1)}$ combinations. The first half of these combinations corresponds to $k_{r} = 0$ and the second half corresponds to $k_{r} = 1$. That is, the first $2^{(q - r)}$ elements of $S(j,.)$ are even and the next $2^{(q - r)}$ elements of $S(j,.)$ are odd. Note that we did not consider the other $r - 1$ entries that show up in the complete dyadic string $\operatorname{dyad(k)}$ of $k$. This is justified by the fact that the pattern we just described will remain true for any configuration of the disregarded $r - 1$ entries, as changing any of these entries allows another total of $2^{(q - r + 1)}$ possible combinations with the same pattern for $\operatorname{dyad_r(k)}$. From this, we infer that $S(j,.)$ is completely defined by alternating sequences of even and odd elements, with each sequence having length $T_{r} = 2^{(q - r)}$. Thus, using that $w_{(j,k)} = (-1)^{S(j,k)}$, the Rademacher row $w_{(j,.)} = w_{(l_{r},.)}$ will be given by
\begin{equation}
    w_{(l_{r},.)} = \bigl[(R_{r})(-R_{r})(R_{r})(-R_{r})\cdots(R_{r})(-R_{r})\bigr],
\end{equation}
with $R_{r} = (+1)(\times T_{r})$ and period $T_r = 2^{(q - r)}$.

\textbf{2. Case \(j \neq l_r\).} Even though $j$ here cannot be written as a power of $2$, it can always be written as a sum of these powers: 
\begin{equation}
    j = \sum_{i = 1}^{h_j}l_{m_i} =\sum_{i = 1}^{h_j}2^{(m_i - 1)},
\end{equation}
and we conveniently choose the indices $m_i$ to match the positions of the non-zero binary elements of $j$, i.e., $j_{m_i} = 1$. That is, if we consider only the non-zero binary elements $\operatorname{bin_{\neq 0}(j)}$ of the binary representation $\operatorname{bin(j)}$ of $j$, then 
\begin{equation} \operatorname{bin_{\neq 0}(j)} = (j_{m_{h_j}}\cdots j_{m_2}j_{m_1}).
\end{equation}
Thus, for any $0 \leq k \leq 2^q - 1$ we have
\begin{equation}
    S(j,k) = \sum_{i = 1}^qj_ik_i = \sum_{i = 1}^{h_j}j_{m_i}k_{m_i}.
\end{equation}
Because each term $j_{m_i}k_{m_i}$ in the sum above is the corresponding $l_{m_i}$ term $S(l_{m_i},k)$, it must be true that
\begin{align}
    w_{(j,k)} &= (-1)^{\sum_{i = 1}^{h_j}S(l_{m_i},k)} \notag\\
    &= \prod_{i = 1}^{h_j}(-1)^{S(l_{m_i},k)} \notag \\
    \label{wjk_prod}
    &= \prod_{i = 1}^{h_j}w_{(l_{m_i},k)}.
\end{align}
Eq. (\ref{wjk_prod}) means that any row $w_{(j,)}$ is the column-wise product of the corresponding Rademacher rows $w_{(l_{m_i},.)}$ such that $j = \sum_{i = 1}^{h_j}l_{m_i}$. 

Since $l_{m_1} < l_{m_2} < \cdots < l_{m_{h_j}}$, the respective periods $T_{m_i} = 2^{(q - m_i)}$, for $1 \leq i \leq h_j$, must obey $T_{m_1} > T_{m_2} > \cdots > T_{m_{h_j}}$. We know that for $0 \leq k \leq T_{m_1} - 1$, we have $w_{(l_{m_1},k)} = +1$. Then, if $0 \leq k \leq T_{m_1} - 1$, the elements $w_{(j,k)}$ of Eq. (\ref{wjk_prod}) can be written as $w_{(j,k)} = \prod_{i = 2}^{h_j}w_{(l_{m_i},k)}$. Now, any Rademacher row $w_{(l_{m_i},.)}$ has a period $T_{m_i}$ and thus is composed of alternating pairs of sequences, where each pair is made of a sequence of $T_{m_i}$ terms all equal to $+1$ and a sequence of $T_{m_i}$ terms all equal to $-1$. This implies that the number of times $L_{[1,i]}$ we can fit these pairs of $2T_{m_i}$ terms into the length of the largest period $T_{m_1}$ is given by
\begin{align}
    L_{[1,i]} &= \frac{T_{m_1}}{2T_{m_i}} \notag\\
    &= \frac{2^{(q - m_1)}}{2 \times 2^{(q - m_i)}} \notag\\
    &=  \frac{2^{(m_i - m_1)}}{2} \notag\\
    &= \frac{l_{m_i}}{2l_{m_1}},
\end{align}
which evidently shows, for $i \geq 2$, that $L_{[1,i]}$ is an integer, specifically, a power of $2$. Now, for $T_{m_1} \leq k \leq 2T_{m_1} - 1$ we must have $w_{(l_{m_1},k)} = -1$ , and the pattern of $w_{(j,.)}$ will occur again for these next $T_{m_1}$ values of $k$, except that in this case we get a sign change, i.e., $w_{(j,k)} = -\prod_{i = 2}^{h_j}w_{(l_{m_i},k)}$. After these first $2T_{m_1}$ terms, the periodicity of $w_{(l_{m_1,.)}}$ implies that the aforementioned pattern will continue to repeat itself until we have the row $w_{(j,.)}$ completely filled. From that, we conclude that $T_{m_1}$ is a common period for every composing Rademacher row of $w_{(j,.)}$. We can also analyze how many pairs of $2T_{m_i}$ terms fit into the length of any other period $T_{m_{i'}} > T_{m_i}$, that is, calculate the expression for $L_{[i',i]}$. In this case, for any $i > i'$, we have
\begin{align}
    L_{[i',i]} &= \frac{T_{m_{i'}}}{2T_{m_i}} \notag\\
    &= \frac{2^{(q - m_{i'})}}{2 \times 2^{(q - m_i)}} \notag\\
    &= \frac{2^{(m_i - m_{i'})}}{2} \notag\\
    &= \frac{l_{m_i}}{2l_{m_{i'}}}.
\end{align}

This shows that any period $T_{m_{i'}}$, for $h_j - 1 \geq i' \geq 1$, has a length that can be perfectly fit by an integer number of pairs of $2T_{m_i}$ terms such that $T_{m_{i'}} > T_{m_i}$. The extreme case of $T_{m_{h_j}}$, however, will never have any lower period, a relevant fact that leads to the following definition:
\begin{definition}
To proceed, for $1 \leq i\leq h_j$, we define the recursive relation
\begin{equation}
\label{recursive}
    R_{m_{(i-1)}} = \bigl[(R_{m_i})(-R_{m_i})(R_{m_i})(-R_{m_i})\cdots(R_{m_i})(-R_{m_i})\bigr],
\end{equation}
where $R_{m_{h_j}} = (+1)(\times T_{m_{h_j}})$, and $R_{m_0}$ and $R_{m_i}$ have respective periods $T_{m_0} = 2^q$ and $T_{m_i} = 2^{(q - m_i)}$.  
\end{definition}
Utilizing the recursive relation in Eq. (\ref{recursive}), we demonstrate that $w_{(j,.)} = R_{m_0}$. This recursive relation maintains a critical connection with the Rademacher rows $w_{(l_{m_i},.)}$. The initial condition $R_{m_{h_j}} = (+1)(\times T_{m_{h_j}})$ is chosen due to $T_{m_{h_j}}$ being the minimal period, thus serving as the recursive sequence's base case. For periods $T_{m_{(i - 1)}} > T_{m_{h_j}}$, the sequence iteratively incorporates $R_{m_i}$, further integrating the effects of $R_{m_{(i + 1)}}$ until the base case $R_{m_{h_j}} = (+1)(\times T_{m_{h_j}})$ is reached. Particularly, $R_{m_0}$ is synthesized through several $(R_{m_1})(-R_{m_1})$ pairs, effectively encapsulating the periodic characteristics of the composing Rademacher rows of $w_{(j,.)}$. Therefore, each term in the sequence $R_{m_1}$ of $T_{m_1}$ terms corresponds to either $\prod_{i = 2}^{h_j}w_{(l_{m_i},k)}$ or its negative counterpart, $-\prod_{i = 2}^{h_j}w_{(l_{m_i},k)}$, for the relevant $T_{m_1}$ columns $k$. As we have shown, the products $\pm \prod_{i = 2}^{h_j}w_{(l_{m_i},k)}$ fully compose the row $w_{(j,.)}$, leading to the conclusion that $w_{(j,.)} = R_{m_0}$. Hence, we have:
\begin{equation}
w_{(j,.)} = \bigl[(R_{m_1})(-R_{m_1})(R_{m_1})(-R_{m_1})\cdots(R_{m_1})(-R_{m_1})\bigr].
\end{equation}
\end{proof}

\section{Demonstration of polynomial cost}
\label{sec:appendixD}

As we have shown in the first Appendix, to use the method of Ref. \cite{welch}, we must calculate the Walsh angles $a_j$, given by
\begin{align}
\label{aj_matrix_form}
    a_j &= \frac{1}{2^q}\sum_{k = 0}^{2^q - 1}f_kw_{jk} \notag \\
    &= \frac{1}{2^q}[w_{(j,.)}] \times [\vec{f}]^{T},
\end{align}
where we have introduced the symbol $T$ for transposition, since $\vec{f}$ is defined as a row vector. The components $f_k$ of $\vec{f}$ are extracted as the eigenvalues of the operator
\begin{equation}
    \hat{f} = \sum_{k = 0}^{2^q - 1}f_k|k\rangle \langle k|.
\end{equation}

Then, a general unitary operator of the form $\hat{U} = e^{i\hat{f}}$, when represented in the computational basis $|k\rangle$, is given by the following diagonal matrix:
\begin{equation}
U =
\begin{bmatrix}
D_{1} & 0 & \cdots & 0 \\
0 & D_{2} & \cdots & 0 \\
\vdots & \vdots & \ddots & \vdots \\
0 & 0 & \cdots & D_{d}
\end{bmatrix},
\end{equation}
where $d^2 = 2^q$, with $d$ a power of $2$, and the $D_s$ are diagonal matrices for $1 \leq s \leq d$. In this paper, we define these diagonal matrices $D_s$ as
\begin{equation}
D_s =
\begin{bmatrix}
e^{i\varphi_s} & 0 & \cdots & 0 \\
0 & e^{2i\varphi_s} & \cdots & 0 \\
\vdots & \vdots & \ddots & \vdots \\
0 & 0 & \cdots & e^{di\varphi_s}
\end{bmatrix},
\end{equation}
where $i^2 = -1$, $\varphi_s = s\kappa$ and $\kappa$ is a real number. Comparing with the unitary gate $\hat{U}(t)$ of Eq. (\ref{unitary}) for our quantum algorithm, we must use $\kappa = -\omega t$. However, since $\kappa$ is a global factor for $\vec{f}$ and thus a multiplicative factor in the expression for $a_j$ in Eq. (\ref{aj_matrix_form}), we can choose $\kappa = 1$ to facilitate further calculations. The vector $\vec{f}$ that we are going to use is then given by
\begin{align}
\label{vec_f}
\vec{f} = & [(1,2,3,\cdots,d),(2,4,6,\cdots,2d), \nonumber \\
& (3,6,9,\cdots,3d),(4,8,12,\cdots,4d), \nonumber \\
& \cdots,(d,2d,3d,\cdots,d^2)].
\end{align}

There is also another global factor $\frac{1}{2^q} = \frac{1}{d^2}$ in the expression (\ref{aj_matrix_form}), which can be equally disregarded in the calculations. Taking into account these two global factors, $a_j$ will be redefined to be
\begin{equation}
    a_j = [w_{(j,.)}] \times [\vec{f}]^{T},
\end{equation}
with $\vec{f}$ given by Eq. (\ref{vec_f}). Thus, the relevant Walsh angles $a_j(t)$ for $\hat{U}(t)$ should be defined as 
\begin{equation}
\label{aj(t)}
a_j(t) = \frac{-\omega t}{d^2}a_j.
\end{equation}

Since there is no risk of confusion, $a_j$ and $a_j(t)$ are both referred to as `Walsh angles' in this work. The following theorem concerns the derivations of the expressions for the non-zero Walsh angles $a_j(t)$ corresponding to the unitary gate $\hat{U}(t)$ of Eq. (\ref{unitary}). Even though we are using $q$ even here, it should be noted that a similar theorem holds for odd values of $q$, opening possibilities for the efficient implementation of any $2^q \times 2^q$ diagonal unitary gate $\hat{U}(t)$ with the form of Eq. (\ref{unitary}). \\

\begin{theorem} \label{theorem2}
Let $W$ be the set of non-zero Walsh angles $a_j(t)$ for the $q$-qubit unitary gate $\hat{U}(t)$ of Eq. (\ref{unitary}), with $q$ even. Then 
\begin{equation}
W = W_1 \cup W_2,
\end{equation}
with
\begin{equation}
W_1 = \left\{ a_j(t) = 
\begin{cases} 
\frac{d(1+d)\omega t}{8j} & \text{if $j \leq d/2$}, \\
\frac{d^2(1+d)\omega t}{8j} & \text{if $j \geq d$},
\end{cases}
\Big| h_j = 1 \right\},
\end{equation}
and
\begin{align}
W_2 = & \Big\{ a_j(t) = \frac{-d^3\omega t}{16l_{m_1}l_{m_2}} \Big| h_j = 2  \\ 
& \text{ and } j = l_{m_1} + l_{m_2} \text{ for } l_{m_1} \leq d/2 \text{ and } l_{m_2} \geq d \Big\}, \nonumber
\end{align}
where $h_j$ is the Hamming weight of $j$ and $d^2 = 2^q$. 
\end{theorem} 

The quantum circuit corresponding to the implementation of $\hat{U}(t)$ is presented for $d = 4$ in Fig. \ref{fig:U_d4_SM}. As it is evident, Theorem \ref{theorem2} eliminates the vast majority of the $2^q - 1$ Walsh angles $a_j(t)$ that would be necessary, in general, to implement this type of unitary gate exactly.  In Lemma \ref{lemma}, we show that the number of gates needed to exactly implement $\hat{U}(t)$ is a polynomial function of $q$.

\begin{figure*}[t]
    \centering    \includegraphics[width=0.8\linewidth]{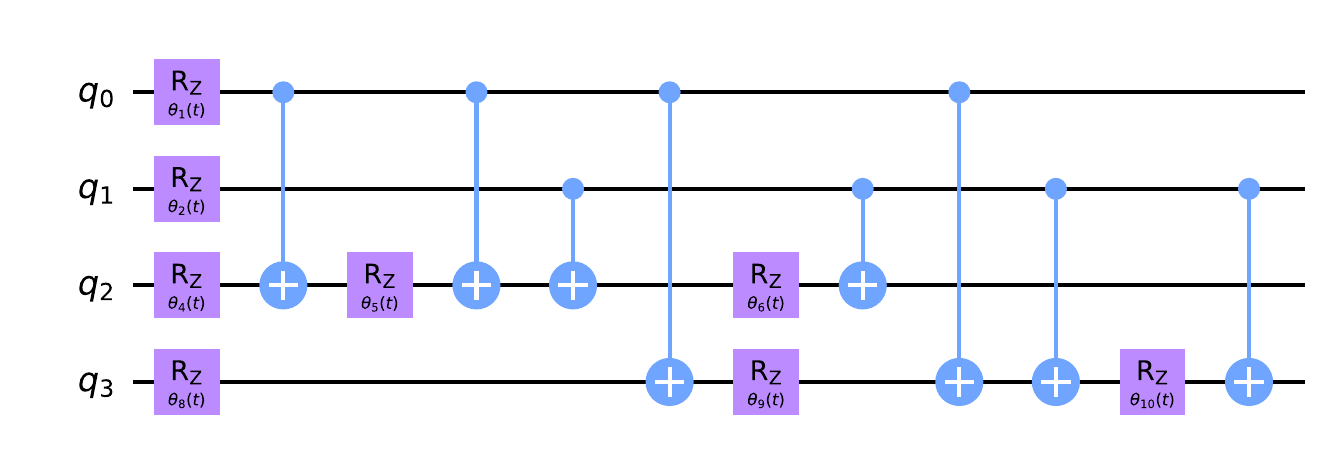}    \captionsetup{justification=raggedright,singlelinecheck=false}
    \caption{Quantum circuit that implements the unitary gate $\hat{U}(t)$ for $d = 4$. We used the mathematical identity demonstrated in Appendix \ref{sec:appendixB} to express an exponential of $Z$ Pauli gates as a single $Z$-rotation $\hat{R}_z(\theta_j(t))$, with the angles given by $\theta_j(t) = -2a_j(t)$, and several controlled-NOT gates.}
    \label{fig:U_d4_SM}
\end{figure*}

In what follows, we list some definitions and their respective properties, in order to use them in the proof of Theorem \ref{theorem2}, that will be presented afterwards.

\begin{definition}
\label{elementary_vector}
Elementary vectors are defined as 
\begin{equation}
\vec{P_{\eta}} = [\eta,2\eta,3\eta,\cdots,d\eta].
\end{equation}
\end{definition}

\begin{definition}
\label{extended_vector}
Extended vectors are defined as 
\begin{equation}
\vec{P}_{\{\alpha|\alpha + \beta\}} = [\vec{P}_{\alpha},\vec{P}_{\alpha+1},\vec{P}_{\alpha+2},\cdots,\vec{P}_{\alpha + \beta}].
\end{equation}
\end{definition}

\begin{definition}
\label{partial_vector}
Partial vectors are defined as
\begin{equation}
\vec{p}^{\,(\eta)}_{\{\alpha|\alpha + \beta\}} = [p^{(\eta)}_{\alpha},p^{(\eta)}_{\alpha + 1},p^{(\eta)}_{\alpha + 2},\cdots,p^{(\eta)}_{\alpha + \beta}],
\end{equation}
where we defined the notation $\vec{P_{\eta}} = [p^{(\eta)}_1,p^{(\eta)}_2,p^{(\eta)}_3,\cdots,p^{(\eta)}_d]$ for elementary vectors.
\end{definition}

\begin{property}
Using Eq. (\ref{vec_f}) for $\vec{f}$ and definition \ref{elementary_vector} for $\vec{P}_\eta$, we can write 
\begin{equation}
\vec{f} = [\vec{P_1},\vec{P_2},\vec{P_3},\cdots,\vec{P_d}].
\end{equation}
\end{property}

\begin{property}
\label{P_linear}
By the definition \ref{elementary_vector} for elementary vectors $\vec{P}_\eta$, it follows that 
\begin{equation}
\vec{P_{\eta}} = \eta \vec{P_1}.
\end{equation}
\end{property}

\begin{property}
\label{P_diff}
It follows directly from property \ref{P_linear} that
\begin{equation}
\vec{P_{\eta}} - \vec{P}_{\eta + \Delta} = -\Delta\vec{P_1}.
\end{equation}
\end{property}

\begin{property}
\label{extended_P_diff}
Applying property \ref{P_diff} for each elementary vector of definition \ref{extended_vector} gives us
\begin{align}
& \vec{P}_{\{\alpha|\alpha+\beta\}} - \vec{P}_{\{\alpha + \Delta|\alpha + \Delta +\beta\}} \nonumber \\
=  &[\vec{P}_{\alpha},\vec{P}_{\alpha+1},\vec{P}_{\alpha+2},\cdots,\vec{P}_{\alpha + \beta}] \nonumber \\
& - [\vec{P}_{\alpha + \Delta},\vec{P}_{\alpha + \Delta+1},\vec{P}_{\alpha + \Delta +2},\cdots,\vec{P}_{\alpha + \Delta + \beta}] \notag\\
=& \big[\vec{P}_{\alpha} - \vec{P}_{\alpha + \Delta},\vec{P}_{\alpha+1} - \vec{P}_{\alpha + \Delta + 1},\vec{P}_{\alpha+2} \nonumber \\
& - \vec{P}_{\alpha + \Delta + 2},\cdots,\vec{P}_{\alpha + \beta} - \vec{P}_{\alpha + \Delta + \beta}\big] \notag\\
=& -\Delta[\vec{P}_{1},\vec{P}_{1},\vec{P}_{1},\cdots,\vec{P}_{1}].
\end{align}
\end{property}

\begin{property}
\label{p_diff}
It follows, respectively, from the definition \ref{partial_vector} of partial vectors $\vec{p}^{\,(\eta)}_{\{\alpha|\alpha + \beta\}}$, property \ref{P_linear}, and definition \ref{elementary_vector} applied to $\vec{P}_1$, that
\begin{align}
& \vec{p}^{\,(\eta)}_{\{\alpha|\alpha + \beta\}} - \vec{p}^{\,(\eta)}_{\{\alpha + \Delta|\alpha + \Delta + \beta\}} \nonumber \\ =& \big[p^{(\eta)}_{\alpha},p^{(\eta)}_{\alpha + 1},p^{(\eta)}_{\alpha + 2},\cdots,p^{(\eta)}_{\alpha + \beta}] \nonumber \\ 
& -[p^{(\eta)}_{\alpha + \Delta},p^{(\eta)}_{\alpha + \Delta + 1},p^{(\eta)}_{\alpha + \Delta + 2},\cdots,p^{(\eta)}_{\alpha + \Delta + \beta}\big] \notag \\
=& \big[p^{(\eta)}_{\alpha} - p^{(\eta)}_{\alpha + \Delta},p^{(\eta)}_{\alpha + 1} - p^{(\eta)}_{\alpha + \Delta + 1},p^{(\eta)}_{\alpha + 2} - p^{(\eta)}_{\alpha + \Delta + 2},\nonumber \\ 
& \cdots,p^{(\eta)}_{\alpha + \beta} - p^{(\eta)}_{\alpha + \Delta + \beta}\big] \notag\\
=&\eta\big[p^{(1)}_{\alpha} - p^{(1)}_{\alpha + \Delta}, p^{(1)}_{\alpha + 1} - p^{(1)}_{\alpha + \Delta + 1},p^{(1)}_{\alpha + 2} - p^{(1)}_{\alpha + \Delta + 2},\nonumber \\ 
& \cdots, p^{(1)}_{\alpha + \beta} - p^{(1)}_{\alpha + \Delta + \beta}] \notag\\
= & -\eta\Delta[1,1,1,\cdots,1].
\end{align}
\end{property}

Building upon the definitions and properties outlined earlier, we now introduce an important vector, denoted by $\vec{f}^{\,\sigma}_{m_i}$, constructed from components of $\vec{f}$. This vector emerges as a cornerstone of our analysis, serving as the basis for deriving key properties that will be extensively used in our proof. Alongside $\vec{f}^{\,\sigma}_{m_i}$, we define a scalar quantity, $F^{\sigma}_i$, designed to facilitate subsequent calculations. 

\begin{definition}
\label{f_partial}
Let $i$ and $\sigma$ be two integers such that $1 \leq i \leq h_j$ and $1 \leq \sigma \leq l_{m_i}$, where $l_{m_i} = \frac{d^2}{2T_{m_i}}$. We define $\vec{f}^{\,\sigma}_{m_i}$ as the vector formed by the $\sigma$-th $2T_{m_i}$ elements of $\vec{f}$. 
\end{definition}

\begin{definition}
\label{F_scalar}
We use the recursive relation (\ref{recursive}) and definition \ref{f_partial} to define the scalar $F^{\sigma}_i = [(R_{m_i})(-R_{m_i})] \times [\vec{f}^{\,\sigma}_{m_i}]^{T}$. 
\end{definition}

\begin{property}
If we use definitions \ref{f_partial} and \ref{F_scalar} with $i = 1$, then
\begin{align}
& a_j = [w_{(j,.)}] \times [\vec{f}]^{T} \notag \\
&= [R_{m_0}] \times [\vec{f}]^{T} \notag \\
&= [(R_{m_1})(-R_{m_1})(R_{m_1})(-R_{m_1})\cdots(R_{m_1})(-R_{m_1})] \nonumber \\ 
& \times \bigl[\vec{f}^{\,1}_{m_1},\vec{f}^{\,2}_{m_1},\vec{f}^{\,3}_{m_1},\cdots,\vec{f}^{\,\sigma}_{m_1},\cdots,\vec{f}^{\,l_{m_1}}_{m_1}\bigr]^{T} \notag\\
&= \sum_{\tau = 1}^{l_{m_1}}[(R_{m_1})(-R_{m_1})] \times [\vec{f}^{\,\tau}_{m_1}]^{T} \notag \\
\label{F^tau_i}
&= \sum_{\tau = 1}^{l_{m_1}}F^{\tau}_1.
\end{align}
\end{property}

The following two properties form the foundation of our strategy to distinguish which Walsh angles $a_j$ are null and which are not. The technique we will use involves determining whether the composing powers $l_{m_i}$ of $j$ possess corresponding periods $T_{m_i} \geq d$ or $T_{m_i} \leq d/2$. Depending on this categorization, we will then employ either extended or partial vector notation.


\begin{property}
\label{Rxf_extended}
Let $T_{m_i} = \mu_i d$, for $\mu_i \geq 1$ a power of $2$. Then, $\vec{f}^{\,\tau}_{m_i} = [\vec{P}_{\{(2\tau-2)\mu_i + 1|(2\tau-1)\mu_i\}},\vec{P}_{\{(2\tau-1)\mu_i + 1|2\tau\mu_i\}}]$ and after using definition \ref{F_scalar} and property \ref{extended_P_diff}:
\begin{align}
F^{\tau}_i =& [(R_{m_i})(-R_{m_i})] \nonumber \\
&\times [{\vec{P}_{\{(2\tau-2)\mu_i + 1|(2\tau-1)\mu_i\}},\vec{P}_{\{(2\tau-1)\mu_i + 1|2\tau\mu_i\}}}]^{T} \notag \\
=& [R_{m_i}] \times [\vec{P}_{\{(2\tau-2)\mu_i + 1|(2\tau-1)\mu_i\}}]^{T} \nonumber \\
&+ [-R_{m_i}] \times [\vec{P}_{\{(2\tau-1)\mu_i + 1|2\tau\mu_i\}}]^{T} \notag \\
=& [R_{m_i}] \nonumber \\
&\times [\vec{P}_{\{(2\tau-2)\mu_i + 1|(2\tau-1)\mu_i\}} - \vec{P}_{\{(2\tau-1)\mu_i + 1|2\tau\mu_i\}}]^{T} \notag\\
\label{FRP}
=& [R_{m_i}] \times [-\mu_i (\vec{P}_1,\vec{P}_1,\vec{P}_1,\cdots,\vec{P}_1)]^{T}.
\end{align}
Thus, by defining $F^{\tau}_1 = F_1$, we have the following property for expression (\ref{F^tau_i}):
\begin{equation}
a_j = l_{m_1}F_1.    
\end{equation}
\end{property}


\begin{property}
\label{Rxf_partial}
Let $T_{m_i} = d/\nu_i$, for $\nu_i \geq 2$ a power of $2$. Then, $\vec{f}^{\,\tau}_{m_i} = [\vec{p}^{\,(\eta)}_{\{(2\tau-2)d/\nu_i + 1|(2\tau-1)d/\nu_i\}},\vec{p}^{\,(\eta)}_{\{(2\tau-1)d/\nu_i + 1|2\tau d/\nu_i\}}]$. After using definition \ref{F_scalar} and property \ref{p_diff}, we obtain
\begin{align}
& F^{\tau}_i = [(R_{m_i})(-R_{m_i})] \nonumber \\
&\times [\vec{p}^{\,(\eta)}_{\{(2\tau-2)d/\nu_i + 1|(2\tau-1)d/\nu_i\}},\vec{p}^{\,(\eta)}_{\{(2\tau-1)d/\nu_i + 1|2\tau d/\nu_i\}}]^{T} \notag \\
&= [R_{m_i}] \times [\vec{p}^{\,(\eta)}_{\{(2\tau-2)d/\nu_i + 1|(2\tau-1)d/\nu_i\}}]^{T} \nonumber \\
&+ [-R_{m_i}] \times [\vec{p}^{\,(\eta)}_{\{(2\tau-1)d/\nu_i + 1|2\tau d/\nu_i\}}]^{T} \notag \\
&= [R_{m_i}] \nonumber \\
&\times [\vec{p}^{\,(\eta)}_{\{(2\tau-2)d/\nu_i + 1|(2\tau-1)d/\nu_i\}} -\vec{p}^{\,(\eta)}_{\{(2\tau-1)d/\nu_i + 1|2\tau d/\nu_i\}}]^{T} \notag\\
\label{FRp}
&= [R_{m_i}] \times [-\frac{\eta d}{\nu_i}(1,1,1,\cdots,1)]^{T}.
\end{align}
We know that $1 \leq \eta \leq d$ and within a fixed $\eta$, there is $\frac{d}{T_{m_1}}$ partial vectors of the form $\vec{p}^{\,(\eta)}_{\{\alpha|\alpha + T_{m_1} - 1\}}$. Then, we define $F^{\tau}_1 = F^{(\eta)}_1$ and replace the sum of Eq. (\ref{F^tau_i}) on $\tau$ with the sum on $\eta$ and a multiplication by the term $\frac{d}{2T_{m_1}} = \frac{\nu_1}{2}$, to obtain the following property:
\begin{equation}
a_j = \frac{\nu_1}{2}\sum_{\eta = 1}^{d}F^{(\eta)}_1.
\end{equation} 
\end{property}

\begin{proof}
Having stated these definitions and properties, now we go to the proof of the theorem. We shall prove separately the formulas of $a_j(t)$ for each Hamming weight $h_j = 1$, $h_j = 2$ and $h_j \geq 3$.

\textbf{1. Case \(h_j = 1\).}
In this case, we have $j = l_{m_1} = 2^{m_1 - 1}$ and $R_{m_1} = (+1)(\times T_{m_1})$. \\

\textbf{1.1. Sub-case \(l_{m_1} \leq d/2\).}
Here, we have $T_{m_1} = \mu_{1} d$, for $\mu_{1} \geq 1$ a power of $2$. Firstly, we calculate the scalar $F_1$ using Eq. (\ref{FRP}):
\begin{align}
    F_1 &= [R_{m_1}] \times [-\mu_1(\vec{P}_1,\vec{P}_1,\vec{P}_1,\cdots,\vec{P}_1)]^{T} \notag \\
    &= [(+1)(\times T_{m_1})] \times [-\mu_1(\vec{P}_1,\vec{P}_1,\vec{P}_1,\cdots,\vec{P}_1)]^{T} \notag \\
    &= -\mu_1 \bigl(P_1 + P_1 + P_1 + \cdots + P_1 \bigr) \notag \\
    &= -\mu_1^{2}P_1,
\end{align}
where we have defined $P_1$ as the sum of all the elements of $\vec{P}_1$, that is, $P_1 = 1 + 2 + 3 + \cdots + d$. Because $P_1  = \frac{(1+d)d}{2}$, then
\begin{equation}
    F_1 = -\frac{\mu_1^{2}(1+d)d}{2}.
\end{equation}

Now, since $T_{m_1} = \mu_1 d$ and by definition $T_{m_1} = d^2/2l_{m_1}$, with $l_{m_1} = j$, we can write $\mu_1 = d/2j$. If we use property \ref{Rxf_extended} to calculate $a_j$, then
\begin{align}
    a_j &= l_{m_1}F_1 \notag \\
    &= j\bigl(-\frac{\mu_1^{2}(1+d)d}{2}\bigr) \notag \\
    &= \frac{-(1+d)d^3}{8j}.
\end{align} 
To obtain the relevant Walsh angle $a_j(t)$, we make use of Eq. (\ref{aj(t)}). Thus
\begin{equation}
    a_j(t) = \frac{d(1+d)\omega t}{8j}.
\end{equation}
\\

\textbf{1.2. Subcase \(l_{m_1} \geq d\).}
For this subcase, we have $T_{m_1} = d/\nu_1$ for $\nu_1 \geq 2$ a power of $2$. The scalar $F^{(\eta)}_1 = F^{\tau}_1$ is calculated by the expression (\ref{FRp}):
\begin{align}
F^{(\eta)}_1 &= [R_{m_1}] \times [-\frac{\eta d}{\nu_1}(1,1,1,\cdots,1)]^{T} \notag \\
&= [(+1)(\times T_{m_1})] \times [-\frac{\eta d}{\nu_1}(1,1,1,\cdots,1)]^{T} \notag \\
&= \frac{-\eta d}{\nu_1}(1 + 1 + 1 + \cdots + 1) \notag \\
&= \frac{-\eta d^2}{\nu_1^{2}}.
\end{align}

Using that $T_{m_1} = d/\nu_1 = d^2/2l_{m_1}$, with $l_{m_1} = j$, we write $\nu_1 = 2j/d$. Then, using property \ref{Rxf_partial} to calculate the Walsh angles $a_j$:
\begin{align}
a_j &= \frac{\nu_1}{2}\sum_{\eta = 1}^{d}F^{(\eta)}_1 \notag \\
&= \frac{-\nu_1 d^2}{2\nu_1^{2}}\sum_{\eta = 1}^{d}\eta \notag \\
&= \frac{-(1+d)d^4}{8j}.
\end{align}
Therefore, by Eq. (\ref{aj(t)}):
\begin{equation}
    a_j(t) = \frac{d^2(1+d)\omega t}{8j}.
\end{equation}
Thus, for $h_j = 1$, we always have
\begin{equation}
    a_j(t) \neq 0.
\end{equation} \\

\textbf{2. Case \(h_j = 2\).}
In this case, we have $j = l_{m_1} + l_{m_2}$. We use the recursive relation of Eq. (\ref{recursive}) to write
\begin{equation}
    R_{m_1} = \bigl[(R_{m_2})(-R_{m_2})(R_{m_2})(-R_{m_2})\cdots(R_{m_2})(-R_{m_2})\bigr],
\end{equation}
with $R_{m_2} = (+1)(\times T_{m_2})$.\\

\textbf{2.1. Sub-case \(l_{m_1} < d/2, l_{m_2} \leq d/2\).}
Here, we have $T_{m_1} = \mu_1 d$ and $T_{m_2} = \mu_2 d$, where both $\mu_1 \geq 2$ and $\mu_2 \geq 1$ are powers of $2$. Firstly, we calculate the scalar $F_1$ using Eq. (\ref{FRP}):
\begin{align}
    F_1 =& [R_{m_1}] \times [-\mu_1(\vec{P}_1,\vec{P}_1,\vec{P}_1,\cdots,\vec{P}_1)]^{T} \notag \\
    =& \bigl[(R_{m_2})(-R_{m_2})(R_{m_2})(-R_{m_2})\cdots(R_{m_2})(-R_{m_2})\bigr] \nonumber \\ 
    & \times [-\mu_1(\vec{P}_1,\vec{P}_1,\vec{P}_1,\cdots,\vec{P}_1)]^{T} \notag \\
    =& \frac{-\mu_1 T_{m_1}}{2T_{m_2}}\bigl([R_{m_2}] \times (\vec{P}_1,\vec{P}_1,\cdots,\vec{P}_1)^{T} - [R_{m_2}] \nonumber \\ 
    & \times (\vec{P}_1,\vec{P}_1,\cdots,\vec{P}_1)^{T}\bigr) \notag \\
    &= 0,
\end{align}
where the factor $T_{m_1}/2T_{m_2}$ was introduced because we collapsed the sum of all the $T_{m_1}/2T_{m_2}$ products $[(R_{m_2})(-R_{m_2})] \times [\vec{P}_1,\vec{P}_1,\cdots,\vec{P}_1]^{T}$ into just a single product. Then, using property \ref{Rxf_extended} for $a_j$, we obtain
\begin{align}
    a_j &= l_{m_1}F_1 \notag \\
    &= 0,
\end{align}
and by Eq. (\ref{aj(t)}):
\begin{equation}
    a_j(t) = 0.
\end{equation} \\

\textbf{2.2. Sub-case \(l_{m_1} \geq d, l_{m_2} > d\).}
In this sub-case, we have $T_{m_1} =  d/\nu_1$ and $T_{m_2} = d/\nu_2$, where both $\nu_1 \geq 2$ and $\nu_2 \geq 4$ are powers of $2$. Now, we calculate the scalar $F^{(\eta)}_1$ using Eq. (\ref{FRp}):
\begin{align}
    F^{(\eta)}_1 =& [R_{m_1}] \times [\frac{-\eta d}{\nu_1}(1,1,1,\cdots,1)]^{T} \notag \\
    =& \bigl[(R_{m_2})(-R_{m_2})(R_{m_2})(-R_{m_2})\cdots(R_{m_2})(-R_{m_2})\bigr] \nonumber \\ 
    & \times [\frac{-\eta d}{\nu_1}(1,1,1,\cdots,1)]^{T} \notag \\
    =& \frac{-\eta d T_{m_1}}{2 \nu_1 T_{m_2}}\bigl([R_{m_2}] \times (1,1,\cdots,1)^{T} - [R_{m_2}] \nonumber \\ 
    & \times (1,1,\cdots,1)^{T}\bigr) \notag \\
    &= 0,
\end{align}
for any $1 \leq \eta \leq d$. From property \ref{Rxf_partial} for $a_j$, we conclude that
\begin{align}
    a_j &= \frac{\nu_1}{2}\sum_{\eta = 1}^{d}F^{(\eta)}_1 \notag \\
    &= 0.
\end{align}
Thus, by Eq. (\ref{aj(t)}), we must also have
\begin{equation}
    a_j(t) = 0.
\end{equation} 

\textbf{2.3. Sub-case \(l_{m_1} \leq d/2, l_{m_2} \geq d\).} For this final sub-case, we have $T_{m_1} = \mu_1 d$ and $T_{m_2} = d/\nu_2$, where $\mu_1 \geq 1$ and $\nu_2 \geq 2$ are powers of $2$. We start by calculating the scalar $F_1$ using Eq. (\ref{FRP}):
\begin{align}
\label{hj2_F1}
    F_1 =& [R_{m_1}] \times [-\mu_1(\vec{P}_1,\vec{P}_1,\vec{P}_1,\cdots,\vec{P}_1)]^{T} \notag \\
    =& \bigl[(R_{m_2})(-R_{m_2})(R_{m_2})(-R_{m_2})\cdots(R_{m_2})(-R_{m_2})\bigr] \nonumber \\ 
    & \times [-\mu_1(\vec{P}_1,\vec{P}_1,\vec{P}_1,\cdots,\vec{P}_1)]^{T}.
\end{align}
However, we can no longer calculate products of the form $[R_{m_2}] \times (\vec{P}_1,\vec{P}_1,\cdots,\vec{P}_1)^{T}$, because now $R_{m_2}$ has length $T_{m_2} = d/\nu_2$, which is less than the length of $\vec{P}_1$. To be able to calculate the right side of expression (\ref{hj2_F1}), we break each $\vec{P}_1$ into smaller parts of length $T_{m_2}$, using partial vectors $\vec{p}^{\,(\eta)}_{\{\alpha|\alpha + T_{m_2} - 1\}}$. Firstly, we notice that $\vec{P}_1$ corresponds to partial vectors with $\eta = 1$. Secondly, since $\vec{P}_1$ has length $d$, we can fit $d/T_{m_2} = \nu_2$ partial vectors $\vec{p}^{\,(\eta)}_{\{\alpha|\alpha + T_{m_2} - 1\}}$ into $\vec{P}_1$. That is, considering Eq. (\ref{FRp}), the $\nu_2/2$ products that we have to calculate are related to $F^{\tau}_2 = F^{(1)}_2$ by the expression 
\begin{align}
\label{mu1F2}
& -\mu_1 F^{(1)}_2 = [(R_{m_2})(-R_{m_2})] \times \nonumber \\ 
& \big[-\mu_1 (\vec{p}^{\,(1)}_{\{(2\tau-2)d/\nu_2 + 1|(2\tau-1)d/\nu_2\}},\vec{p}^{\,(1)}_{\{(2\tau-1)d/\nu_2 + 1|2\tau d/\nu_2\}})\big]^{T} \notag \\
&=  [R_{m_2}] \times [\frac{\mu_1 d}{\nu_2}(1,1,\cdots,1)]^{T}.
\end{align}

After calculating the product of Eq. (\ref{mu1F2}) and multiplying it by the number $\nu_2/2$, we should then multiply the result by the number of elementary vectors $\vec{P}_1$ appearing in Eq. (\ref{hj2_F1}), which is $T_{m_1}/d = \mu_1$. Therefore
\begin{align}
    F_1 &= (\mu_1)(\frac{\nu_2}{2}) (-\mu_1 F^{(1)}_2) \notag \\
    &= [R_{m_2}] \times [\frac{\mu_1^{2} d}{2}(1,1,\cdots,1)]^{T} \notag \\
    &= [(+1)(\times T_{m_2})] \times [\frac{\mu_1^{2} d}{2}(1,1,\cdots,1)]^{T} \notag \\
    &= \frac{\mu_1^{2} d}{2}(1 + 1 + \cdots + 1) \notag \\
    &= \frac{\mu_1^{2} d^2}{2\nu_2}.
\end{align}

With $F_1$ calculated, $a_j$ will be given by property \ref{Rxf_extended}. To simplify the final result for $a_j$, we will use $T_{m_1} = \mu_1 d = d^2/2l_{m_1}$ and $T_{m_2} = d/\nu_2 = d^2/2l_{m_2}$ to write $\mu_1 = d/2l_{m_1}$ and $\nu_2 = 2l_{m_2}/d$. Therefore, by property \ref{Rxf_extended} we have that
\begin{align}
    a_j &= l_{m_1}F_1 \notag \\
    &= l_{m_1}\bigl(\frac{\mu_1^{2} d^2}{2\nu_2}\bigr) \notag \\
    &= \frac{d^5}{16l_{m_1}l_{m_2}}.
\end{align}
Again, by making use of Eq. (\ref{aj(t)}), we obtain
\begin{equation}
    a_j(t) = \frac{-d^3\omega t}{16l_{m_1}l_{m_2}}.
\end{equation}
Thus, for $h_j = 2$, if $l_{m_1} \leq d/2$ and $l_{m_2} \geq d$, we necessarily have 
\begin{equation}
a_j(t) \neq 0,    
\end{equation}
otherwise $a_j(t) = 0$. 

\textbf{3. Case \(h_j \geq 3\).}
For this case, we have in general $j = l_{m_1} + l_{m_2} + \sum_{i = 3}^{h_j}l_{m_i}$. 

\textbf{3.1. Sub-case \(l_{m_1} < d/2, l_{m_2} \leq d/2\).}
Here, we have a similar situation to sub-case $2.1$: $T_{m_1} = \mu_1 d$ and $T_{m_2} = \mu_2 d$. However, we also have other powers $l_{m_i}$ with respective periods $T_{m_i} = d^2/2l_{m_i}$ for $i \geq 3$. We will show that for any such $T_{m_i}$, it must be true that $a_j = 0$. We recall the recursive relation (\ref{recursive}) and use Eq. (\ref{FRP}) to obtain
\begin{align}
& F_1 = [R_{m_1}] \times [-\mu_1(\vec{P}_1,\vec{P}_1,\vec{P}_1,\cdots,\vec{P}_1)]^{T} \notag \\
    &= \bigl[(R_{m_2})(-R_{m_2})(R_{m_2})(-R_{m_2})\cdots(R_{m_2})(-R_{m_2})\bigr] \nonumber \\ 
    & \times [-\mu_1(\vec{P}_1,\vec{P}_1,\vec{P}_1,\cdots,\vec{P}_1)]^{T} \notag \\
    &= \frac{-\mu_1 T_{m_1}}{2T_{m_2}}\bigl([R_{m_2}] \times (\vec{P}_1,\vec{P}_1,\cdots,\vec{P}_1)^{T} - [R_{m_2}] \nonumber \\ 
    & \times (\vec{P}_1,\vec{P}_1,\cdots,\vec{P}_1)^{T}\bigr) \notag \\
    &= 0.
\end{align}
In the scenario of sub-case $2.1$, we have $R_{m_2} = (+1)(\times T_{m_2})$ and $F_1 = 0$. Although for the present sub-case we have $R_{m_2} \neq (+1)(\times T_{m_2})$, the scalar $F_1$ is again identically null, just like in sub-case $2.1$. Note that we have not imposed any restrictions on $l_{m_i}$, that is, $F_1 = 0$ regardless of whether $l_{m_i} \leq d/2$ or $l_{m_i} \geq d$. Now, from property \ref{Rxf_extended} we have $a_j = l_{m_1}F_1$. Thus, if there are any others $l_{m_i}$ such that $j = l_{m_1} + l_{m_2} + \sum_{i = 3}^{h_j}l_{m_i}$, it holds that
\begin{equation}
    a_j = 0.
\end{equation} 
Therefore, by Eq. (\ref{aj(t)}), it follows that:
\begin{equation}
    a_j(t) = 0.
\end{equation}

\textbf{3.2. Sub-case \(l_{m_1} \geq d, l_{m_2} > d, l_{m_i} > d\).} This situation is an extension of sub-case $2.2$: $T_{m_1} = d/\nu_1$ and $T_{m_2} = d/\nu_2$. Again, we have powers $l_{m_i}$ with respective periods $T_{m_i}$, with $i \geq 3$. By the recursive relation (\ref{recursive}) and Eq. (\ref{FRp}), we obtain
\begin{align}
   & F^{(\eta)}_1 = [R_{m_1}] \times [\frac{-\eta d}{\nu_1}(1,1,1,\cdots,1)]^{T} \notag \\
    &= \bigl[(R_{m_2})(-R_{m_2})(R_{m_2})(-R_{m_2})\cdots(R_{m_2})(-R_{m_2})\bigr] \nonumber \\ 
    & \times [\frac{-\eta d}{\nu_1}(1,1,1,\cdots,1)]^{T} \notag \\
    &= \frac{-\eta d T_{m_1}}{2 \nu_1 T_{m_2}}\bigl([R_{m_2}] \times (1,1,\cdots,1)^{T} - [R_{m_2}] \nonumber \\ 
    & \times (1,1,\cdots,1)^{T}\bigr) \notag \\
    &= 0,
\end{align}
for any $1 \leq \eta \leq d$. Similarly to what happened for $F_1$ in sub-case $3.1$, here the result of $F^{(\eta)}_1$ is independent of the form of $R_{m_2}$. This is relevant because in the present sub-case, $R_{m_2}$ is obtained by others $R_{m_i}$ that we did not specify. Now, from property \ref{Rxf_partial} we have $a_j = \frac{\nu_1}{2}\sum_{\eta = 1}^{d}F^{(\eta)}_1$. Then, it must be true that
\begin{equation}
    a_j = 0,
\end{equation}
and by Eq. (\ref{aj(t)}) it follows that
\begin{equation}
    a_j(t) = 0.
\end{equation} 

\textbf{3.3. Sub-case \(l_{m_1} \leq d/2, l_{m_2} \geq d, l_{m_i} > d\).}
In sub-case $2.3$, we have shown that $a_j \neq 0$. Here, since there are other powers $l_{m_i}$, the situation, however, will turn out to be different. We should repeat the calculations of sub-case $2.3$ by first using Eq. (\ref{mu1F2}):
\begin{align}
    -\mu_1 F^{(1)}_2 = [R_{m_2}] \times [\frac{\mu_1 d}{\nu_2}(1,1,\cdots,1)]^{T}. 
\end{align}
We must calculate $F_1$, since $T_{m_1} \geq d$. From the calculations of sub-case $2.3$, $F_1$ is given by $F_1 = (\mu_1)(\frac{\nu_2}{2})(-\mu_1 F^{(1)}_2)$. Thus
\begin{align}
    & F_1 = (\mu_1)(\frac{\nu_2}{2}) (-\mu_1 F^{(1)}_2) \notag \\
    &= [R_{m_2}] \times [\frac{\mu_1^{2} d}{2}(1,1,\cdots,1)]^{T} \notag \\
    &= \bigl[(R_{m_3})(-R_{m_3})(R_{m_3})(-R_{m_3})\cdots(R_{m_3})(-R_{m_3})\bigr] \nonumber \\ 
    & \times [\frac{\mu_1^{2} d}{2}(1,1,\cdots,1)]^{T} \notag \\
    &= \frac{\mu_1^{2}dT_{m_2}}{4T_{m_3}}\bigl([R_{m_3}] \times (1,\cdots,1)^{T} - [R_{m_3}] \times (1,\cdots,1)^{T}\bigr) \notag \\
    &= 0.
\end{align}
From property \ref{Rxf_extended}, the Walsh angles for this sub-case are given by $a_j = l_{m_1}F_1$. Therefore
\begin{equation}
    a_j = 0.
\end{equation}
And from Eq. (\ref{aj(t)}), we get
\begin{equation}
    a_j(t) = 0.
\end{equation}
Just like it happened for sub-case $2.2$, where we had $l_{m_1} \geq d$ and $l_{m_2} > d$, the Walsh angles $a_j(t)$ are also null here. This is true for any $l_{m_{i \geq 3}}$ that composes $j$. The conclusion is that for any $j$ with $h_j \geq 3$, we have $a_j(t) = 0$. 

We know that the only non-zero Walsh angles $a_j(t)$ are those with Hamming weight $h_j = 1$ or $h_j = 2$, that is, cases $1$ and $2$. In case $1$, it is always true that $a_j(t) \neq 0$. In case $2$, the Walsh angles are $a_j(t) \neq 0$ if and only if we simultaneously have $l_{m_1} \leq d/2$ and $l_{m_2} \geq d$. Therefore, if $W$ is the set of non-zero Walsh angles, then it is the union of two subsets $W_1$ and $W_2$, composed, respectively, by the non-zero Walsh angles with $h_j = 1$ and $h_j = 2$. We can summarize this as
\begin{equation}
W = W_1 \cup W_2,
\end{equation}
with
\begin{equation}
W_1 = \left\{ a_j(t) = 
\begin{cases} 
\frac{d(1+d)\omega t}{8j} & \text{if $j \leq d/2$}, \\
\frac{d^2(1+d)\omega t}{8j} & \text{if $j \geq d$},
\end{cases}
\Big| h_j = 1 \right\},
\end{equation}
and
\begin{align}
W_2 = & \Big\{ a_j(t) = \frac{-d^3\omega t}{16l_{m_1}l_{m_2}} \Big| h_j = 2 \text{ and } j = l_{m_1} + l_{m_2}, \nonumber \\ 
& \text{where } l_{m_1} \leq d/2 \text{ and } l_{m_2} \geq d \Big\}.
\end{align}
\end{proof}
\begin{lemma}
\label{lemma}
Let $G_2(q)$ be the number of gates necessary to implement the $q$-qubit unitary gate $\hat{U}(t)$ exactly. If we use only $Z$-rotations and controlled-NOT gates, then $G_2(q) = \frac{3}{4}q^2 + q$.
\end{lemma}

\begin{proof}
As it is established in Appendix \ref{sec:appendixB}, we can calculate the exponential operators $e^{ia_j(t) \hat{w}_j}$ in Eq. (\ref{walsh_unitary}) by applying a $Z$-rotation on qubit $q_{m_{h_j}}$, where $m_{h_j}$ is the MSB of $j$, and two controlled-NOT gates targeted on $q_{m_{h_j}}$ for each controlling qubit. That is, the number of gates for a single $e^{ia_j(t) \hat{w}_j}$ is given by one $Z$-rotation and $2(h_j - 1)$ controlled-NOT gates, resulting in $2h_j - 1$ gates. Now, from Theorem \ref{theorem2}, the only non-zero Walsh angles $a_j(t)$ are those for which we have $h_j = 1$ or $h_j = 2$. 
    
For $h_j = 1$ it is always true that $a_j(t) \neq 0$ and the respective values of $j$ correspond to powers of $2$. We know that within $1 \leq j \leq 2^q - 1$, there are $q$ powers of $2$. Then, the total number of gates $G^{(1)}_2(q)$ necessary for $h_j = 1$ is $G^{(1)}_2(q) = (2h_j - 1)q$. Thus
    \begin{equation}
     G^{(1)}_2(q) = q.
    \end{equation}
    
For $h_j = 2$, we have $a_j(t) \neq 0$ if and only if $j = l_{m_1} + l_{m_2}$ with $l_{m_1} \leq d/2$ and $l_{m_2} \geq d$. To find how many gates $G^{(2)}_2(q)$ are needed here, we must count how many combinations of $l_{m_1}$ and $l_{m_2}$ are possible. Firstly, we notice that $d/2 = 2^{\frac{q}{2}}/2 = 2^{(\frac{q}{2} - 1)}$. If we pick $l_{m_1}$ as satisfying $l_{m_1} = d/2$, then $d/2 = l_{m_1} = 2^{(m_1 - 1)}$ implies that $m_1 = \frac{q}{2}$. That is, there are $\frac{q}{2}$ possible values for $l_{m_1} \leq d/2$. Because there is a total of $q$ powers of $2$ in the interval $1 \leq j \leq 2^q - 1$, then there is also $q - \frac{q}{2} = \frac{q}{2}$ possible values for $l_{m_2}$. We conclude that the number of combinations of $l_{m_1}$ and $l_{m_2}$ is $\bigl(\frac{q}{2}\bigr)\bigl(\frac{q}{2}\bigr) = \frac{q^2}{4}$. Then, the number $G^{(2)}_2(q)$ of gates necessary for $h_j = 2$ is $G^{(2)}_2 = (2h_j - 1)\frac{q^2}{4}$. Thus
    \begin{equation}
        G^{(2)}_2 = \frac{3}{4}q^2.
    \end{equation}
    
    Finally, the total gate cost $G_2(q) = G^{(1)}_2(q) + G^{(2)}_2(q)$ for the implementation of the unitary $\hat{U}(t)$ is
    \begin{equation}
        G_2(q) = \frac{3}{4}q^2 + q.
    \end{equation}
\end{proof}

\section{Purity estimation using a variation of the SWAP test}
\label{sec:appendixE}

In Ref. \cite{ekert}, the authors explore an interferometric setup to extract $\operatorname{Tr}(\hat{U}\hat{\rho})$ based on its correlation with the visibility $v$, where $\hat{U}$ represents a unitary gate and $\hat{\rho}$ denotes the density operator of the system. Their investigation draws parallels between this quantum circuit and the one employed in the SWAP test. However, there are notable distinctions: they utilize a controlled-$\hat{U}$ gate instead of a controlled-SWAP gate and consider density operators $\hat{\rho}$ instead of pure states $|\psi\rangle$. 

The authors further contend that by selecting $\hat{U}$ as the SWAP gate and letting $\hat{\rho} = \hat{\rho}_{S_1} \otimes \hat{\rho}_{S_2}$ represent the joint density operator of two subsystems $S_1$ and $S_2$, a specific scenario arises:
\begin{equation} \label{TrS1S2}     \operatorname{Tr}(\hat{\rho}_{S_1} \hat{\rho}_{S_2}) = 2P_0 - 1, \end{equation}
where $P_0$ signifies the probability of measuring the state $|0\rangle$ for the ancilla qubit subsequent to the application of the unitary gates forming the quantum circuit.

In this Appendix, we present an operational proof for Eq. (\ref{TrS1S2}). Particularly, when $\hat{\rho}_{S_1} = \hat{\rho}_{S_2}$, it yields the purity $\operatorname{Tr}(\hat{\rho}^2_{S_1}) = 2P_0 - 1$, a key quantity for our quantum algorithm. Before delving into the proof, we first introduce an identity pertaining to the SWAP gate.

\begin{proposition}
Let $\hat{V}_1$ and $\hat{V}_2$ be two linear operators acting, respectively, on $d^2$-dimensional Hilbert spaces $H_1$ and $H_2$, with $H_1 = H_2$. Then, the following identity for the $\widehat{SWAP}$ gate holds:
\begin{equation}
\label{swap_id}
    \operatorname{Tr}\bigl((\hat{V}_1 \otimes \hat{V}_2)\widehat{SWAP}\bigr) = \operatorname{Tr}(\hat{V}_1\hat{V}_2).
\end{equation}
\end{proposition}

\begin{proof}
We will calculate both sides of Eq. (\ref{swap_id}) and show that they lead to the same expression. To do that, we start by defining the matrix representations of $\hat{V}_i$ on the computational basis:
\begin{equation}
    \hat{V}_i = 
\begin{bmatrix}
V^{(1,1)}_i & V^{(1,2)}_i & \cdots & V^{(1,d)} \\
V^{(2,1)}_i & V^{(2,2)}_i & \cdots & V^{(2,d)}_i \\
\vdots & \vdots & \ddots & \vdots \\
V^{(d,1)}_i & V^{(d,2)}_i & \cdots & V^{(d,d)}_i
\end{bmatrix}.
\end{equation}
Then, we can write
\begin{equation}
    \hat{V}_1\hat{V}_2 =
\begin{bmatrix}
\sum_{k = 1}^dV^{(1,k)}_1V^{(k,1)}_2 & \cdots & * \\
\vdots & \ddots & \vdots \\
* & * & \sum_{k = 1}^dV^{(d,k)}_1V^{(k,d)}_2
\end{bmatrix},
\end{equation}
with $*$ denoting matrix elements we do not need. Thus
\begin{equation}
    \operatorname{Tr}(\hat{V}_1\hat{V}_2) = \sum_{j, k = 1}^dV^{(j,k)}_1V^{(k,j)}_2.
\end{equation}

Now, we calculate the left side of Eq. (\ref{swap_id}):
\begin{align}
    & \operatorname{Tr}\Big(\Big(\hat{V}_1 \otimes \hat{V}_2\Big)\widehat{SWAP}\Big) \nonumber \\
    &= \operatorname{Tr}\Big(\Big(\sum_{j,k = 1}^d V^{(j,k)}_1 |j\rangle \langle k| \otimes \sum_{l,m = 1}^d V^{(l,m)}_2 |l\rangle \langle m|\Big)\widehat{SWAP}\Big) \notag \\
    &= \sum_{j,k = 1}^d \sum_{l,m = 1}^d V^{(j,k)}_1 V^{(l,m)}_2 \operatorname{Tr}\bigl(|j\rangle \otimes |l\rangle \langle k| \otimes \langle m|\widehat{SWAP}\bigr) \notag \\
    &= \sum_{j,k = 1}^d \sum_{l,m = 1}^d V^{(j,k)}_1 V^{(l,m)}_2 \operatorname{Tr}\bigl(|j\rangle \otimes |l\rangle \langle m| \otimes \langle k|\bigr) \notag \\
    &= \sum_{j,k = 1}^d \sum_{l,m = 1}^d V^{(j,k)}_1 V^{(l,m)}_2 \langle m|j\rangle \langle k|l\rangle \notag \\
    &= \sum_{j,k = 1}^d V^{(j,k)}_1 V^{(k,j)}_2.
\end{align}
Therefore
\begin{equation}
\operatorname{Tr}\bigl((\hat{V}_1 \otimes \hat{V}_2)\widehat{SWAP}\bigr) = \operatorname{Tr}(\hat{V}_1\hat{V}_2).    
\end{equation}
\end{proof}

\begin{figure}[H]
    \centering
    \includegraphics[width=1.0\linewidth]{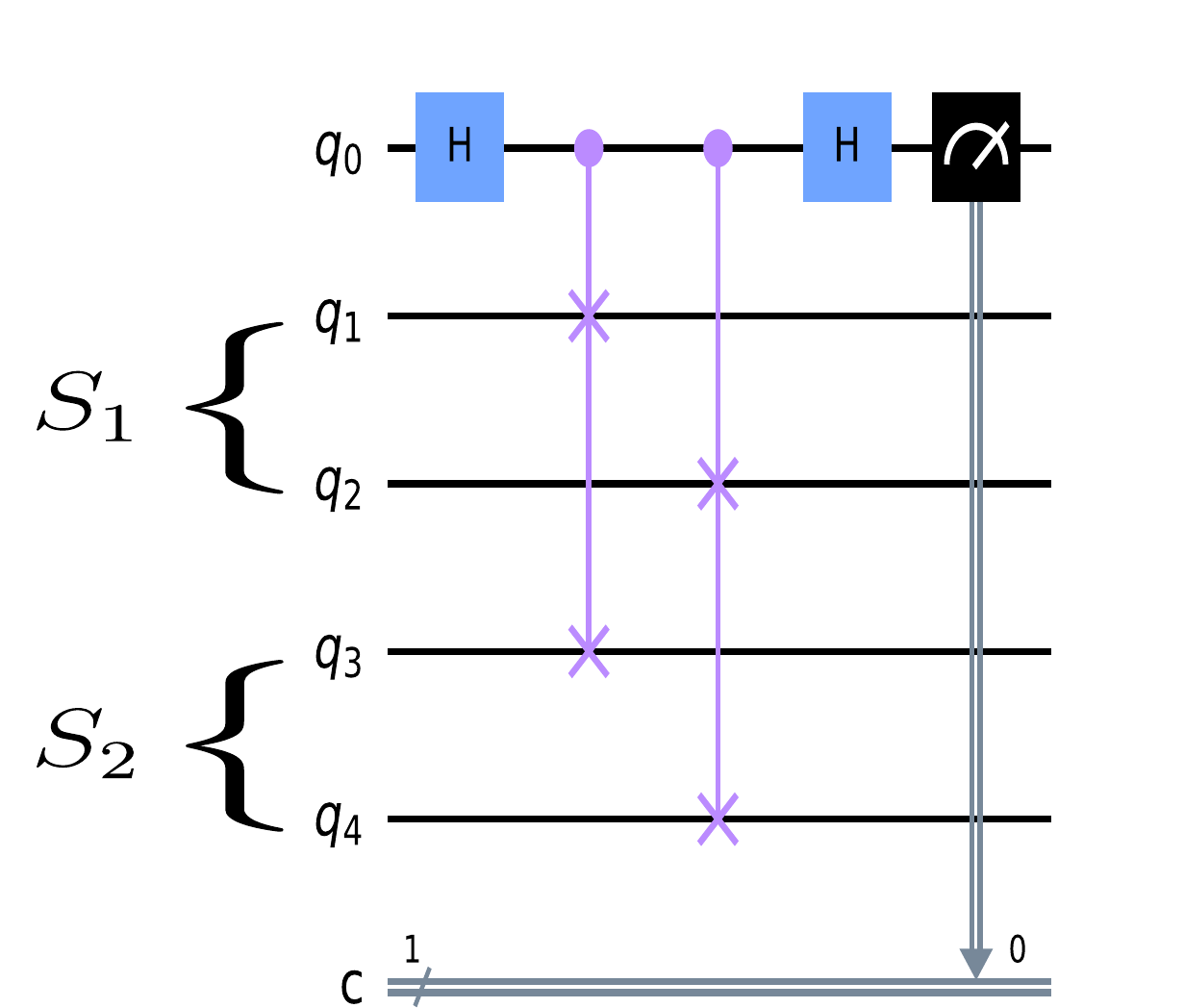}   \captionsetup{justification=raggedright,singlelinecheck=false}
\caption{A schematic representation of a modified SWAP test employed to determine $\operatorname{Tr}(\hat{\rho}_{S_1}\hat{\rho}_{S_2})$ in a four-qubit system ($q=4$). Initially, qubits $q_1$ and $q_2$ are prepared in the state $\hat{\rho}_{S_1}$, while qubits $q_3$ and $q_4$ are prepared in the state $\hat{\rho}_{S_2}$. An ancilla qubit $q_0$ is introduced. The circuit involves a Hadamard gate applied to $q_0$, qubit-qubit controlled-SWAP gates between the two qubit sets, and a final Hadamard gate on $q_0$. By measuring $q_0$ multiple times and obtaining $P_0$, the quantity $\operatorname{Tr}(\hat{\rho}_{S_1}\hat{\rho}_{S_2})$ can be estimated as $2P_0 - 1$.}
    \label{fig:swap_general_SM}
\end{figure}

We are now prepared to demonstrate the validity of Eq. (\ref{TrS1S2}). This proof will be conducted by constructing the proposed quantum circuit introduced in Ref. \cite{ekert}, illustrated in Fig. \ref{fig:swap_general_SM} for the specific scenario involving $q = 4$ qubits and one ancilla qubit. Our system comprises an ancilla qubit $q_0$ and two subsystems, denoted as $S_1$ and $S_2$. Initially, the system's density operator $\hat{\rho}$ is given by the tensor product:
\begin{equation}
\hat{\rho} = |0\rangle \langle 0| \otimes \hat{\rho}_{S_1} \otimes \hat{\rho}_{S_2}.
\end{equation}

Following the sequence of gates outlined in the quantum circuit, we define $\hat{\rho}^{(1)} = \hat{H} \hat{\rho} \hat{H}$, $\hat{\rho}^{(2)} = \bigl(\widehat{CSWAP}\bigr)\bigl(\hat{\rho}^{(1)}\bigr)\bigl(\widehat{CSWAP}\bigr)$ and $\hat{\rho}^{(3)} = \hat{H}\hat{\rho}^{(2)}\hat{H}$. Commencing with $\hat{\rho}^{(1)}$, we proceed by applying a Hadamard gate $\hat{H}$ to $q_0$:
\begin{align}
\hat{\rho}^{(1)} &= \hat{H} \hat{\rho} \hat{H} \notag \\
&= \frac{1}{2}(|0\rangle + |1\rangle)(\langle 0| + \langle 1|) \otimes \hat{\rho}_{S_1} \otimes \hat{\rho}_{S_2} \notag \\
&= |+\rangle \langle +| \otimes \hat{\rho}_{S_1} \otimes \hat{\rho}_{S_2},
\end{align}
with $|\pm\rangle := \frac{1}{\sqrt{2}}(|0\rangle \pm |1\rangle)$. By applying the controlled-SWAP gate $\widehat{CSWAP}$, we obtain
\begin{align}
& \hat{\rho}^{(2)} = \bigl(\widehat{CSWAP}\bigr)\bigl(\hat{\rho}^{(1)}\bigr)\bigl(\widehat{CSWAP}\bigr)\notag \\
=& \bigl(|0\rangle \langle 0| \otimes \hat{I} + |1\rangle \langle 1| \otimes \widehat{SWAP}\bigr)\bigl(|+\rangle \langle +| \otimes \hat{\rho}_{S_1} \otimes \hat{\rho}_{S_2}\bigr) \nonumber \\ 
& \times \bigl(|0\rangle \langle 0| \otimes \hat{I} + |1\rangle \langle 1| \otimes \widehat{SWAP} \bigr)\notag \\
=& \frac{1}{\sqrt{2}}\bigl(|0 \rangle \langle +| \otimes \hat{\rho}_{S_1} \otimes \hat{\rho}_{S_2} + |1 \rangle \langle +| \otimes \widehat{SWAP} (\hat{\rho}_{S_1} \otimes \hat{\rho}_{S_2})\bigr) \nonumber \\ & \times \bigl(|0\rangle \langle 0| \otimes \hat{I} + |1\rangle \langle 1| \otimes \widehat{SWAP} \bigr) \notag \\
=& \frac{1}{2}\bigl(|0 \rangle \langle 0| \otimes \hat{\rho}_{S_1} \otimes \hat{\rho}_{S_2} + |1 \rangle \langle 0| \otimes \widehat{SWAP} (\hat{\rho}_{S_1} \otimes \hat{\rho}_{S_2})  \nonumber \\ 
& + |0 \rangle \langle 1| \otimes (\hat{\rho}_{S_1} \otimes \hat{\rho}_{S_2}) \widehat{SWAP} \notag \\
& + |1 \rangle \langle 1| \otimes \widehat{SWAP} (\hat{\rho}_{S_1} \otimes \hat{\rho}_{S_2}) \widehat{SWAP}\bigr).
\end{align}
To finalize, we apply another Hadamard gate to $q_0$:
\begin{align}
& \hat{\rho}^{(3)} = \hat{H}\hat{\rho}^{(2)}\hat{H} \notag \\
=& \frac{1}{2}\bigl(|+ \rangle \langle +| \otimes \hat{\rho}_{S_1} \otimes \hat{\rho}_{S_2} + |- \rangle \langle +| \otimes \widehat{SWAP} (\hat{\rho}_{S_1} \otimes \hat{\rho}_{S_2})  \nonumber \\ 
& + |+ \rangle \langle -| \otimes (\hat{\rho}_{S_1} \otimes \hat{\rho}_{S_2}) \widehat{SWAP} \notag \\
& + |- \rangle \langle -| \otimes \widehat{SWAP} (\hat{\rho}_{S_1} \otimes \hat{\rho}_{S_2}) \widehat{SWAP}\bigr).
\end{align}

The next step is to calculate the reduced density operator $\hat{\rho}_0$ for the ancilla qubit. To do that, we take the partial trace over $S_1$ and $S_2$:
\begin{align}
    & \hat{\rho}_0 = \operatorname{Tr}_{S_1 S_2}(\hat{\rho}^{(3)}) \notag \\
    =& \frac{1}{2}\big(|+\rangle \langle +| \operatorname{Tr}(\hat{\rho}_{S_1} \otimes \hat{\rho}_{S_2}) + |- \rangle \langle +|\operatorname{Tr}(\widehat{SWAP} (\hat{\rho}_{S_1} \otimes \hat{\rho}_{S_2}))  \nonumber \\ 
    & + |+ \rangle \langle -| \operatorname{Tr} ((\hat{\rho}_{S_1} \otimes \hat{\rho}_{S_2}) \widehat{SWAP}\big) \notag \\
    &+ |- \rangle \langle -| \operatorname{Tr}(\widehat{SWAP} (\hat{\rho}_{S_1} \otimes \hat{\rho}_{S_2}) \widehat{SWAP})\bigr) \notag \\
    =& \frac{1}{2}\bigl((|+ \rangle \langle +| + | - \rangle \langle - |) + (|- \rangle \langle +| + |+ \rangle \langle -|)  \nonumber \\ 
    & \times \operatorname{Tr}((\hat{\rho}_{S_1} \otimes \hat{\rho}_{S_2})\widehat{SWAP})\bigr) \notag \\
    =& \frac{\hat{I} + \operatorname{Tr}(\hat{\rho}_{S_1}\hat{\rho}_{S_2})\hat{Z}}{2},
\end{align}
where in the last step we used the identity (\ref{swap_id}). Thus, the probability $P_0 = \operatorname{Tr}(|0\rangle \langle 0|\hat{\rho}_0)$ of obtaining state $|0\rangle$ for the ancilla qubit is
\begin{align}
P_0 &= \operatorname{Tr}\bigl(|0\rangle \langle 0|(\frac{\hat{I} + \operatorname{Tr}(\hat{\rho}_{S_1}\hat{\rho}_{S_2})\hat{Z}}{2})\bigr) \notag \\
&= \frac{1}{2}\operatorname{Tr}\bigl(|0\rangle \langle 0| + \operatorname{Tr}(\hat{\rho}_{S_1}\hat{\rho}_{S_2})|0\rangle \langle 0|\bigr) \notag \\
&= \frac{1 + \operatorname{Tr}(\hat{\rho}_{S_1}\hat{\rho}_{S_2})}{2}.
\end{align}

\begin{figure}[t]
    \centering
    \includegraphics[width=0.8\linewidth]{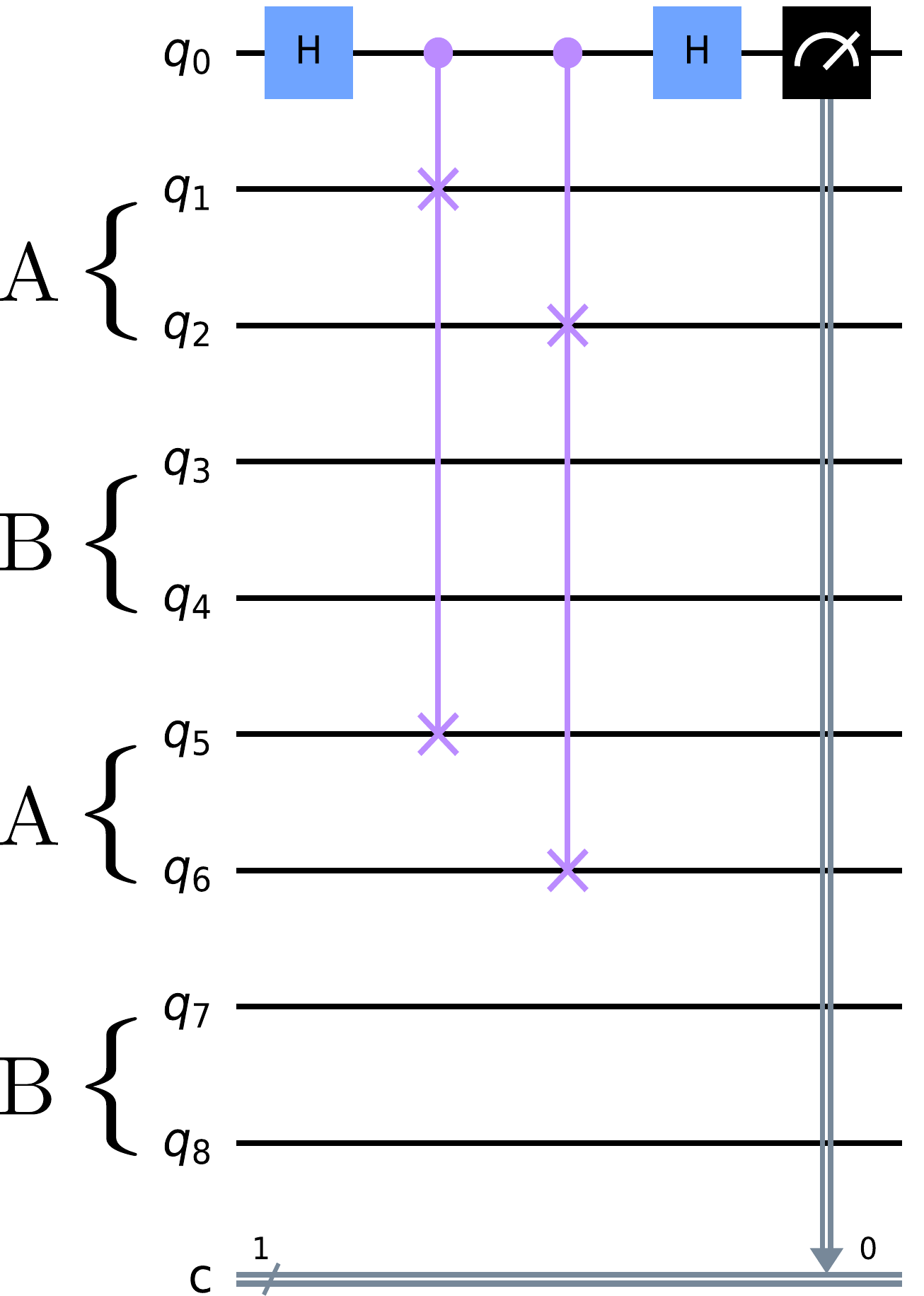}    \captionsetup{justification=raggedright,singlelinecheck=false}
    \caption{Quantum circuit used for the estimation of the reduced purity $\gamma_A$ of a system $AB$ composed by $q = 4$ qubits. This circuit is analogous to the general case, with some exceptions. Here, the first $q$ qubits after the ancilla qubit $q_0$ and the last $q$ qubits are prepared in the same state $\hat{\rho}$. Also, the controlled-SWAP gates are applied only between the first half of qubits of each copy. Then, measuring $q_0$ for various identical circuits allows us to estimate the reduced purity $\gamma_A = \operatorname{Tr}(\hat{\rho}_{A}^2)$, which corresponds to the bipartite reduced state of qubits $q_1$ and $q_2$ or $q_5$ and $q_6$.}
    \label{fig:purity_d4_SM}
\end{figure}

Therefore, we find that $\operatorname{Tr}(\hat{\rho}_{S_1}\hat{\rho}_{S_2}) = 2P_0 - 1.$ Now, specializing to the case where $\hat{\rho}_{S_1} = \hat{\rho}_{S_2}$, we obtain the purity:
\begin{equation}
    \gamma_{S_1} = \operatorname{Tr}(\hat{\rho}^2_{S_1}) = 2P_0 - 1.
\end{equation}
In our quantum algorithm, $\gamma_{S_1} = \gamma_A$ is a function of time and is actually the reduced purity of subsystem $A$. The quantum circuit for this special case is shown in Fig. \ref{fig:purity_d4_SM} for two identical systems $A$ and $B$ with $q = 4$ qubits each.



\begin{thebibliography}{10}
\bibliographystyle{apsrev4-2}

\bibitem{Bressoud1989} D. M. Bressoud, \textit{Factorization and Primality Testing} (Springer-Verlag, New York, 1989).

\bibitem{crandall} R. Crandall and C. Pomerance, \textit{Prime Numbers: A Computational Perspective} (Springer, New York, 2005).

\bibitem{Baillie2021} R. Baillie, A. Fiori, and S. S. Wagstaff Jr., Strengthening the Baillie-PSW primality test, Math. Comp. 90, 1931 (2021).

\bibitem{Granville2005} A. Granville, It is easy to determine whether a given integer is prime, Bull. Amer. Math. Soc. 42, 3 (2005).

\bibitem{schmayer} D. Schumayer and D. A. W. Hutchinson, Colloquium: Physics of the Riemann hypothesis, Rev. Mod. Phys. 83, 307 (2011).

\bibitem{wolf} M. Wolf, Will a physicist prove the Riemann hypothesis?, Rep. Prog. Phys. 83, 036001 (2020).

\bibitem{Feiler2013} C. Feiler and W. P. Schleich, Entanglement and analytical continuation: an intimate relation told by the Riemann zeta function, New J. Phys. 15, 063009 (2013).

\bibitem{Sierra2008} G. Sierra and P. K. Townsend, Landau Levels and Riemann Zeros, Phys. Rev. Lett. 101, 110201 (2008).

\bibitem{aaronson} S. Aaronson, The Prime Facts: From Euclid to AKS, \url{https://www.scottaaronson.com/writings/prime.pdf} (2003).

\bibitem{AKS} M. Agrawal, N. Kayal, and N. Saxena, PRIMES is in P, Annals of Mathematics 160, 781 (2004).

\bibitem{Sieve2017} H. A. Helfgott, An improved sieve of Eratosthenes, Math. Comp. 89, 333 (2017).

\bibitem{Miller1976} G. L. Miller, Riemann's Hypothesis and Tests for Primality, J. Comput. Syst. Sci. 13, 300 (1976).

\bibitem{Adleman1983} L. M. Adleman, On Distinguishing Prime Numbers from Composite Numbers, Annals of Mathematics 117, 173 (1983).

\bibitem{Donis-Vela2017} A. Donis-Vela and J. C. Garcia-Escartin, A quantum primality test with order finding, Quantum Inf. Comp. 18, 1143 (2018).

\bibitem{Chau1997} H. F. Chau and H.-K. Lo, Primality Test Via Quantum Factorization, Int. J. Mod. Phys. C 08, 131 (1997).

\bibitem{Li2012} J. Li, X. Peng, J. Du, and D. Suter, An Efficient Exact Quantum Algorithm for the Integer Square-free Decomposition Problem, Sci. Rep. 2, 1 (2012).

\bibitem{Garcia-Martin2020} D. Garc\'ia-Mart\'in, E. Ribas, S. Carrazza, J. I. Latorre, and G. Sierra, The Prime state and its quantum relatives, Quantum 4, 371 (2020).

\bibitem{Mussardo2020} G. Mussardo, A. Trombettoni, and Z. Zhang, Prime Suspects in a Quantum Ladder, Phys. Rev. Lett. 125, 240603 (2020).

\bibitem{alexandre} A. L. M. Southier, L. F. Santos, P. H. S. Ribeiro, and A. D. Ribeiro, Identifying primes from entanglement dynamics, Phys. Rev. A 108, 042404 (2023).

\bibitem{Bullock2004} S. S. Bullock and I. L. Markov, Asymptotically optimal circuits for arbitrary n-qubit diagonal computations, Quantum Info. Comput. 4, 27 (2004).

\bibitem{welch} J. Welch, D. Greenbaum, S. Mostame, and A. Aspuru-Guzik, Efficient quantum circuits for diagonal unitaries without ancillas, New J. Phys. 16, 033040 (2014).

\bibitem{Bennett1996} C. H. Bennett, H. J. Bernstein, S. Popescu, and B. Schumacher, Concentrating Partial Entanglement by Local Operations, Phys. Rev. A 53, 2046 (1996).

\bibitem{Vidal1999} G. Vidal and R. Tarrach, Robustness of entanglement, Phys. Rev. A 59, 141 (1999).

\bibitem{Basso2022} M. L. W. Basso and J. Maziero, Entanglement monotones from complementarity relations, J. Phys. A: Math. Theor. 55, 355304 (2022).

\bibitem{Scherer2021} M. V. Scherer and A. D. Ribeiro, Entanglement dynamics of spins using a few complex trajectories, Phys. Rev. A 104, 042222 (2021).

\bibitem{ekert} A. K. Ekert, C. M. Alves, D. K. L. Oi, M. Horodecki, P. Horodecki, and L. C. Kwek, Direct Estimations of Linear and Nonlinear Functionals of a Quantum State, Phys. Rev. Lett. 88, 217901 (2002).

\bibitem{Press2007} W. H. Press, S. A. Teukolsky, W. T. Vetterling, and B. P. Flannery, \textit{Numerical Recipes: The Art of Scientific Computing} (Cambridge University Press, New York, 2007).

\bibitem{Qiskit} A. Javadi-Abhari et al., Quantum computing with Qiskit, arXiv:2405.08810 (2024). doi: 10.48550/arXiv.2405.08810..

\bibitem{Nielsen2000} M. A. Nielsen and I. L. Chuang, \textit{Quantum Computation and Quantum Information} (Cambridge University Press, Cambridge, 2000).

\bibitem{Maziero2017} J. Maziero, Computing partial traces and reduced density matrices, Int. J. Mod. Phys. C 28, 1750005 (2017).

\bibitem{Arfken2013} G. B. Arfken, H. J. Weber, and F. E. Harris, \textit{Mathematical Methods for Physicists: A Comprehensive Guide}, 7nd ed. (Elsevier, Oxford, 2013).

\bibitem{Walsh1923} J. L. Walsh, A Closed Set of Normal Orthogonal Functions, Am. J. Math. 45, 5 (1923).

\bibitem{Fine1949} N. J. Fine, On the Walsh Functions, Trans. Am. Math. Soc. 65, 372 (1949).

\bibitem{Zhihua1983} L. Zhihua and Z. Qishan, Ordering of Walsh Functions, IEEE Trans. Electromagn. Compat. 25, 115 (1983).

\bibitem{Yuen1975} C.-K. Yuen, Function Approximation by Walsh Series, IEEE Trans. Comp. 24, 590 (1975).

\bibitem{buhrman} H. Buhrman, R. Cleve, J. Watrous, and R. de Wolf, Quantum Fingerprinting, Phys. Rev. Lett. 87, 167902 (2001).

\bibitem{barenco} A. Barenco, A. Berthiaume, D. Deutsch, A. Ekert, R. Jozsa, and C. Macchiavello, Stabilization of Quantum Computations by Symmetrization, SIAM J. Comput. 26, 1541 (1997).



















\end{thebibliography}
\end{document}